\documentclass[english]{iopart}

\usepackage[T1]{fontenc}
\usepackage[utf8]{inputenc}

\expandafter\let\csname equation*\endcsname\relax %chktex 41
\expandafter\let\csname endequation*\endcsname\relax %chktex 41
\usepackage{amsmath}

\usepackage{amsbsy,amstext,amsthm,amsfonts}

\newcommand{\1}{{\rm 1\hspace{-0.9mm}l}}

\newtheorem{theorem}{Theorem}
\newtheorem{lemma}{Lemma}
\newtheorem{definition}{Definition}

\newtheorem{problem}{Problem}

\usepackage{tikz}%for the figures
\usetikzlibrary{calc,fadings,backgrounds,shapes.geometric,decorations.pathreplacing,positioning,fit,patterns}

\usepackage[colorlinks=true,breaklinks=true]{hyperref}
\usepackage{xspace}
\usepackage{braket}
\renewcommand\bra[1]{{\langle{#1}|}}
\makeatletter
\renewcommand\ket[1]{%
  \@ifnextchar\bra{\k@t{#1}\!}{\k@t{#1}}%
}
\newcommand\k@t[1]{{|{#1}\rangle}}
\makeatother

\renewcommand{\vec}[1]{\boldsymbol{#1}}
\newcommand{\estim}[1]{\hat{#1}}
\newcommand{\adj}[1]{#1^\dagger}
\newcommand{\abs}[1]{\left\vert #1 \right\vert} %chktex 1
\newcommand{\norm}[1]{\left\Vert #1 \right\Vert} %chktex 1

\newcommand{\mineig}{\mathrm{mineig}\,}
\newcommand{\maxeig}{\mathrm{maxeig}\,}

\newcommand{\dd}{\mathrm{d}}
\newcommand{\ii}{\mathrm{i}}
\newcommand{\ee}{\mathrm{e}}

\newcommand{\Eps}{\mathcal{E}}
\newcommand{\States}{\mathcal{S}^+}
\newcommand{\HermTrace}{\mathcal{S}}

\newcommand{\Nhalf}{\tfrac{N}{2}}

\newcommand{\Prob}{\mathbb{P}}

\newcommand{\Vol}{\mathrm{Vol}}
\newcommand{\BR}{\mathbb{R}}
\newcommand{\BC}{\mathbb{C}}
\newcommand{\CR}{\mathcal{C}}

\newcommand{\red}[1]{{\color{red} #1}}

\makeatother

\begin{document}

\title[Error regions in quantum state tomography]{Error regions in quantum state tomography: computational complexity caused by geometry of quantum states}

\author{Daniel Suess,
{\L}ukasz Rudnicki, and
David Gross}
\address{Institute for Theoretical Physics, University of Cologne, Germany}

\author{Thiago O.\ Maciel}
\address{Departamento de F\'{\i}sica -- ICEx -- Universidade Federal de Minas Gerais, Brazil}
\begin{abstract}
	The outcomes of quantum mechanical measurements are inherently random.
	It is therefore necessary to develop stringent methods for quantifying the degree of statistical uncertainty about the results of quantum experiments.
	For the particularly relevant task of quantum state tomography, it has been shown that a significant reduction in uncertainty can be achieved by taking the positivity of quantum states into account.
	However -- the large number of partial results and heuristics notwithstanding -- no efficient general algorithm is known that produces an optimal uncertainty region from experimental data, while making use of the prior constraint of positivity.
	Here, we provide a precise formulation of this problem and show that the general case is NP-hard.
	Our result leaves room for the existence of efficient approximate solutions, and therefore does not in itself imply that the practical task of quantum uncertainty quantification is intractable.
	However, it does show that there exists a non-trivial trade-off between optimality and computational efficiency for error regions.
	We prove two versions of the result: One for frequentist and one for Bayesian statistics.
\end{abstract}

\maketitle

\section{Introduction}
\label{sec:intro}

The outcomes of quantum mechanical measurements are subject to intrinsic randomness.
As a result, all information we obtain about quantum mechanical systems are subject to statistical uncertainty.
It is thus necessary to develop stringent methods for quantifying the degree of uncertainty.
These allow one to decide whether an observed feature can be trusted to be real, or whether it may have arisen from mere statistical fluctuations (a ``fluke'').

In this paper, we concentrate on uncertainty quantification for quantum state tomography (QST)\footnote{%
  This subsumes the more general problem of standard quantum \emph{process} tomography, by way of the Choi-Jamiolkowski isomorphism~\cite{Nielsen_2010_Quantum,Jezek_2003_Quantum,Altepeter_2003_AncillaAssisted}.
}
Here, the task is to infer a density matrix $\varrho_0$, associated with a preparation procedure of a finite-dimensional quantum system, from the outcomes of measurements on independent copies of the system.
In addition to an estimate $\hat \varrho$ for the unknown true state $\varrho_0$, a tomography procedure should rigorously quantify the remaining statistical uncertainty.

We note that QST is an established experimental tool -- in particular in quantum information-inspired setups.
It has been used to characterize quantum states in a large number of different platforms -- Refs.~\cite{Brien_2004_Quantum,Lundeen_2009_Tomography,Terriza_2004_Triggered,Karpinski_2008_FiberOptic,Rippe_2008_Experimental,Steffen_2006_Measurement,Childress_2006_Coherent,Riebe_2007_Quantum,Schwemmer_2014_Experimental} are an incomplete list.

From a technical point of view, uncertainty quantification in QST may seem to be straight-forward.
After choosing a measurement to perform (i.e.\ by specifying a POVM), the probability distribution over the outcomes is a linear function of the unknown state $\varrho_0$.
Inference and uncertainty quantification in linear models are well-studied problems of mathematical statistics.
What makes the QST problem special is the additional constraint that the density matrix $\varrho_0$ be positive semi-definite (\emph{psd}) and have unit-trace.
This \emph{shape constraint} can lead to a significant reduction in uncertainty -- in particular if the true state $\varrho_0$ is close to the boundary of state space:
In this case, it is plausible that a large fraction of possible estimates that seem compatible with the observations
can be discarded, as they lie outside of state space.

Indeed, it is known that taking the psd constraint into account can result in a dramatic -- even unbounded -- reduction in uncertainty.
Prime examples are results that employ positivity to show that even \emph{informationally incomplete} measurements can be used to identify a quantum state with arbitrarily small error~\cite{Cramer_2010_Efficient,Gross_2010_Quantum,Gross_2011_Recovering,Flammia_2012_Quantum,Nickl_2013_Confidence,Kalev_2015_Quantum}.
More precisely, these papers describe ways to rigorously bound the size of a confidence region for the quantum state based only on the observed data and on the knowledge that the data comes from measurements on a valid quantum state.
While these uncertainty bounds can always be trusted without further assumptions, only in very particular situations have they been proven to actually become small.
These situations include the cases where the true state $\varrho_0$ is of low rank~\cite{Flammia_2012_Quantum,Nickl_2013_Confidence}, or admits an economical description as a matrix-product state~\cite{Cramer_2010_Efficient}.

It stands to reason that there are further cases -- not yet identified -- for which the size of an error region can be substantially reduced simply by taking into account the quantum shape constraints.
This motivates the research program this paper is part of: \emph{Understand the general impact of positivity constraints on uncertainty quantification in QST}.

The positive results cited above notwithstanding, it is not obvious how to take the a priori information of positive semi-definiteness into account algorithmically.
The fact that no practical and optimal general-purpose algorithm for quantum uncertainty quantification has been identified could either reflect a limit in our current understanding -- or it could indicate that no efficient algorithm for this problem exists.

In this work, we present first evidence that optimal quantum uncertainty quantification is algorithmically difficult.
We give rigorous notions of optimality both from the point of view of Bayesian statistics (where this concept is fairly canonic) and of ``orthodox'' statistics (where some choices have to be made).
We exhibit special cases for which there does exist an efficient algorithm that identifies optimal error regions.
However, our main result proves that in general, finding these regions is NP-hard and thus computationally intractable.
By working under assumptions that render the unconstrained problem tractable, we show that computational intractability arises solely due to the quantum constraints and not due to the general difficulties of high-dimensional statistics.

The present results do not by themselves imply that the practical problem of uncertainty quantification is unfeasible.
For applications, ``almost-optimal'' regions would be completely satisfactory.
And indeed, a number of techniques for tackling this problem in theory and practice have been proposed (e.g.\ based on \emph{sample splitting}, \emph{resampling}, or on approximations for Bayesian posterior distributions -- c.f.\ Sec.~\ref{sub:intro.related}).
Each of these methods is known analytically or from numerical experiments to perform well in some regimes.
However, this paper does establish that \emph{there is a non-trivial trade-off between optimality and computational efficiency in quantum uncertainty quantification}.
What is more, our work might help guide future efforts that aim to design efficient and optimal estimators: With a very natural construction proven not to possess an efficient algorithm in general, it is now clear that researchers must focus on approximations that circumvent our hardness results.
In general, we hope that this work establishes a framework for future positive and negative results, which will eventually allow us to understand which performance can be achieved.

The rest of this paper is structured as follows. In the subsections below, we comment on use-cases of full QST for high-dimensional quantum systems, summarize related works, and clarify the (non-trivial) issue of ``optimality'' in uncertainty quantification.
We then establish the main result for orthodox statistics in Section~\ref{sec:orthodox} and follow up with a Bayesian treatment in Section~\ref{sec:bayesian}.

\subsection{The need for full tomography}
\label{sub:tomography}

A large number of tomography experiments for quantum systems with hundreds of dimensions has been published, e.g.~\cite{Haeffner_2005_Scalable}.
However, it is not completely obvious that this approach will continue to make sense as dimensions scale up further.

Indeed, a variety of theoretical tools for \emph{quantum hypothesis testing}, \emph{certification}, and scalar \emph{quantum parameter estimation}~\cite{Donnell_2015_Quantum,Audenaert_2008_Asymptotic,Guehne_2009_Entanglement,Flammia_2011_Direct,Schwemmer_2015_Systematic,Li_2016_Optimal} have been developed in the past years, that avoid the costly step of full QST.
Examples include entanglement witnesses~\cite{Guehne_2009_Entanglement} and \emph{direct fidelity estimation}~\cite{Flammia_2011_Direct}.

However, there remain use cases that necessitate full-fledged QST.
We see a particularly important role in the emergent field of \emph{quantum technologies}:
Any technology requires means of certifying that components function as intended and, should they fail to do so, identify the way in which they deviate from the specification.

As an example, consider the implementation of a quantum gate that is designed to act as a component of a universal quantum computing setup.
One could use a certification procedure -- direct fidelity estimation, say -- to verify that the implementation is sufficiently close to the theoretical target that it meets the stringent demands of the quantum error correction threshold.
If it does, the need for QST has been averted.
However, should it fail this test, the certification methods give no indication \emph{in which way} it deviated from the intended behavior.
They yield no actionable information that could be used to adjust the preparation procedure.
The pertinent question ``what went wrong'' cannot be cast as a hypothesis test.

Thus, while many estimation and certification schemes can -- and should -- be formulated without resorting to full tomography, the above example shows that QST remains an important primitive.

\subsection{Error Regions}
\label{sub:intro.error_regions}

As inference based on empirical data is one of the main topics of statistics, it is natural to apply the established notions of uncertainty quantification to QST.
These are either \emph{confidence regions} in orthodox statistics~\cite{Kiefer_2012_Introduction} or \emph{credible regions} in Bayesian statistics~\cite{Bolstad_2007_Introduction}.
The two approaches give rise to different techniques, but most importantly, have very distinct interpretations~\cite{Jaynes_1976_Confidence}.\\

In orthodox (or frequentist) statistics, the task of parameter estimation can be summarized as follows:
We assume that the observed data is generated from a parametric model with true parameter $\Theta$, which is unknown.
From a finite number of observations $X_1, \ldots, X_N$, we must construct an estimate $\hat\Theta$ that should be ``close to'' the true value $\Theta$ in some sense.
The function that maps data to such an estimate is called a (point) estimator.
A \emph{confidence region $\CR$ with coverage $\alpha$} is a region estimator -- that is a function that maps observed data to a subset of the parameter space -- such that the true parameter is contained within it with probability greater than $\alpha$
\begin{equation}
  \label{eq:intro.confidence_region}
  \Prob_\Theta(\CR(X_1,\ldots,X_N) \owns \Theta) \ge \alpha.
\end{equation}
Note that the defining property of a confidence region concerns the behavior of the random function $\CR$ over the course of many (hypothetical) repetitions of the experiment.
No statement is made about a single run.

Of course, Eq.~\eqref{eq:intro.confidence_region} does not uniquely determine a confidence region; it does not even guarantee a sensible quantification of uncertainty, as $\CR$ equal to the whole parameter space fulfills this condition trivially.
Therefore, we consider confidence regions that perform well with respect to (w.r.t.) some notion of optimality:
In general, smaller regions should be preferred since they convey more confidence in the estimate and exclude more alternatives.
But since the size -- as measured by volume -- of a confidence region may depend on the particular data sample as well as the true value of the parameter, different notions of optimality have been introduced~\cite{Pfanzagl_1994_Parametric}.\\

Bayesian statistics on the other hand treats the parameter $\Theta$ itself as a random variable.
The distribution over $\Theta$ reflects our knowledge about the parameters~\cite{Bolstad_2007_Introduction}.
Ahead of observing any data, one has to choose a \emph{prior distribution}, which represents our a priori beliefs.
The observed data is then incorporated using Bayes' rule to \textit{update} the distribution yielding the posterior $\Prob(\Theta | X_1, \ldots, X_N)$.
A credible region $\CR$ (we denote both confidence and credibility regions by the same letter) with credibility $\alpha$ is defined as a subset of the parameter space containing at least mass $\alpha$ of the posterior
\begin{equation}
  \label{eq:intro.credible_region}
  \Prob(\Theta \in \CR | X_1, \ldots, X_N) \ge \alpha.
\end{equation}
In contrast to the orthodox setting, here, the data is assumed to be fixed and the probability is assigned w.r.t.\ $\Theta$.

Since the posterior distribution is uniquely defined by the choice of prior and the data,
there is less ambiguity in the choice of a notion of optimality:
The most natural choice are minimal-volume credible regions.
In case the posterior has the probability density $\pi(\theta)$ w.r.t.\ the volume measure, these are given by regions of highest posterior density
\begin{equation}
  \label{eq:intro.highest_posterior}
  \CR = \{ \theta\colon \pi(\theta) \geq \lambda \},
\end{equation}
where $\lambda$ is determined by the saturation of the credibility level condition~\eqref{eq:intro.credible_region}.

\subsection{Positivity of quantum states}
\label{sub:intro.positivity}

When attempting to construct optimal error regions for QST, we should exploit the physical constraints at hand in order to reduce their size and, therefore, make them more powerful:
every valid density matrix $\varrho$ -- apart from being Hermitian and normalized -- must be positive semidefinite (psd). More formally, in a $d$-dimensional scenario it is required that
\begin{equation}
  \varrho \in \States = \{ \varrho \in \BC^{d \times d}\colon \varrho = \varrho^\dagger, \tr \varrho = 1, \varrho \ge 0 \}.
\end{equation}
Here, $\States$ denotes the set of valid mixed quantum states, which is a proper subset of the real vector space $\HermTrace$ of Hermitian matrices with unit trace.

While the first two properties (hermiticity and normalization) are linear constraints and therefore easy to take into account by virtue of an appropriate parametrization, positivity is far more challenging to employ constructively.
A prime example where this structural information is crucial in the construction of optimal error regions is the application of compressed sensing techniques to QST~\cite{Gross_2010_Quantum,Flammia_2012_Quantum,Carpentier_2015_Uncertainty}.
Compressed sensing allows to recover a low-rank state from informationally incomplete measurements.
Without further assumptions, this can lead to unbounded error regions -- c.f.\ the discussion of Pauli designs in~\cite{Carpentier_2015_Uncertainty} and Sec.~\ref{sub:ortho.optimal}.
Nevertheless, the constraints implied by physical states allow for the construction of confidence regions in this setting~\cite{Carpentier_2015_Uncertainty}, that are of finite size and that become arbitrarily small as the individual measurement errors tend to zero.

However, as the cited work is specifically tailored to the compressed sensing scenario, it is not clear how to extend  it to the general setting of QST.
The purpose of this work is to explore the degree to which positivity can be taken into account in general, if one assumes that computational power is bounded.

\subsection{State of the art}
\label{sub:intro.related}

In practice (e.g.~\cite{Haeffner_2005_Scalable}), uncertainty quantification for tomography experiments is usually based on general-purpose resampling techniques such as ``bootstrapping''~\cite{Efron_1994_Introduction}.
A common procedure is this: For every fixed measurement setting, several repeated experiments are performed.
This gives rise to an empirical distribution of outcomes for this particular setting.
One then creates a number of simulated data sets by sampling randomly from a multinomial distribution with parameters given by the empirical values.
Each simulated data set is mapped to a quantum state using maximum likelihood estimation.
The variation between these reconstructions is then reported as the uncertainty region.
There is no indication that this procedure grossly misrepresents the actual statistical fluctuations.
However, it seems fair to say that its behavior is not well-understood.
Indeed, it is simple to come up with pathological cases in which the method would be hopelessly optimistic:
E.g.\ one could estimate the quantum state by performing only one repetition each, but for a large number of randomly chosen settings.
The above method would then spuriously find a variance of zero.\\

On the theoretical side, some techniques to compute rigorously defined error bars for quantum tomographic experiments have been proposed in recent years.
The works of Blume-Kohout~\cite{Kohout_2012_Robust} as well as Christandl, Renner, and Faist~\cite{Christandl_2012_Reliable,Faist_2015_Practical} exhibit methods for constructing confidence regions for QST based on likelihood level sets.
While very general, neither paper provides a method that has both a runtime guarantee and also adheres to some notion of non-asymptotic optimality~\cite{Kiefer_2012_Introduction,Cam_2012_Asymptotic}.

Some  authors have proposed a ``sample-splitting'' approach, where the first part of the data is used to construct an estimate of the true state, whereas the second part serves to construct an error region around it~\cite{Flammia_2012_Quantum} (based on~\cite{Flammia_2011_Direct}), as well as~\cite{Carpentier_2015_Uncertainty}.
These approaches are efficient, but rely on specific measurement ensembles (operator bases with low operator norm), approach optimality only up to poly-logarithmic factors, and -- in the case of~\cite{Flammia_2012_Quantum, Flammia_2011_Direct} -- rely on adaptive measurements.

Regarding Bayesian methods, the \emph{Kalman filtering} techniques of~\cite{Audenaert_2008_Asymptotic} provide a efficient algorithm for computing credible regions.
This is achieved by approximating all Bayesian distributions over density matrices by Gaussians and restricting attention to ellipsoidal credible regions.
The authors develop a heuristic method for taking positivity constraints into account -- but the degree to which the resulting construction deviates from being optimal remains unknown.
A series of recent papers aim to improve this construction by employing the \emph{particle filter} method for Bayesian estimation and uncertainty quantification
\cite{Granade_2016_Practicala,Wiebe_2015_Bayesian,Ferrie_2014_High}.
Here, Bayesian distributions are approximated as superpositions of delta distributions and credible regions constructed using Monte Carlo sampling. These methods lead to fast algorithms and are more flexible than Kalman filters with regard to modelling prior distributions that may not be well-approximated by any Gaussian. However, once more, there seems to be no rigorous estimate for how far the estimated credible regions deviate from optimality.
Finally, the work in~\cite{Shang_2013_Optimal} constructs optimal credible regions w.r.t.\ a different notion of optimality:
Instead of penalizing sets with larger volume, they aim to minimize the prior probability as suggested by~\cite{Evans_2006_Optimally}.

\section{Orthodox Confidence Regions}
\label{sec:orthodox}

In this section we are going to present the first major result of this work concerned with orthodox confidence regions in QST.
Optimal confidence regions for such high-dimensional parameter estimation problems are quite intricate even without any constraints on the allowed parameters.
There are only few elementary settings, where optimal error regions are known and easily characterized.

Since the goal of this work is to demonstrate that quantum shape constraints severely complicate even ``classically'' simple confidence regions, in the further discussion we restrict the discussion to a simplified setting:
We focus on confidence ellipsoids for Gaussian distributions, which are one of the few easily characterizable examples.
Furthermore, by local asymptotic normality, these arise as a natural approximation in the limit of many measurements.
As we show in the following, even characterizing these highly simplifying ellipsoids with the quantum constraints taken into account constitutes a hard computational problem.
On the other hand, as indicated in the introduction, these structural assumptions may help to reduce the uncertainty tremendously.
Therefore, our work can be interpreted a trade-off between computational efficiency and statistical optimality in QST.

\subsection{Optimal confidence regions for quantum states}
\label{sub:ortho.optimal}

As already indicated in Sec.~\ref{sub:intro.positivity}, the additional information that the true quantum state $\varrho_0$ must belong to the set of positive semidefinite matrices $\States \subset \HermTrace$ can be exploited to possibly improve any confidence region for QST.
This is especially clear for notions of optimality with a loss function stated in terms of volume\footnote{%
  Throughout this work, the volume is taken with respect to the flat Hilbert-Schmidt measure on $\HermTrace$.
}
$\Vol(\cdot)$, as we will show in this section.

We consider an especially simple procedure to take the positivity constraints into account, namely truncating all non-positive matrices from tractable confidence regions for the unconstrained problem.
This approach is mainly motivated by the goal to show that the computational intractability exclusively stems from the quantum constraints and is not caused by difficulties of high-dimensional statistics in general.
Furthermore, we prove in Lemma~\ref{lem:ortho.admissible_truncation} that some notions of optimality, e.g.\ admissibility, are preserved under truncation.
In other words, there are notions of optimality such that truncation of an optimal confidence region for the unconstrained problem gives rise to an optimal region for the constrained one.\\

First, let us introduce the notion of admissibility as given by~\cite[Def. 2.2]{Joshi_1969_Admissibility}.
\begin{definition}\label{def:ortho.admissible}
  A confidence region $\CR$ for the parameter estimation of $\varrho_0 \in \HermTrace$ is called (weakly) \emph{admissible} if there is no other confidence region $\CR'$ that fulfills
  \begin{enumerate}
    \item(equal or smaller volume) $\Vol(\CR'(\vec y)) \le \Vol(\CR(\vec y))$ for almost all observations $y \in \BR^m$
    \item(same or better coverage) $\Prob(\CR' \owns \varrho_0) \ge \Prob(\CR \owns \varrho_0)$ for all $\varrho_0 \in \HermTrace$
    \item(strictly better) strict inequality holds for one $\varrho_0 \in \HermTrace$ in (ii) or on a set of positive measure in (i).
  \end{enumerate}
\end{definition}
In words, $\CR$ is admissible if there is no other confidence region $\CR'$ that performs at least as good as $\CR$ and strictly better for some settings.
The conditions in Def.~\ref{def:ortho.admissible} are stated only for ``almost all'' $\vec y$, since one can always modify the region estimators on sets of measure zero without changing their statistical performance.
A different approach is to state condition (i) in terms of the expected volume\footnote{%
  Here, the average is taken over to the obtained data for a fixed true state $\varrho_0$.
}
, which leads to the notion of strong admissibility~\cite[Def.~7.1]{Joshi_1969_Admissibility}.

Def.~\ref{def:ortho.admissible} can also be stated for the parameter estimation with physical constraints, i.e.\ when $\varrho_0 \in \States$.
The question is: Can we obtain admissible confidence regions $\CR^+\subset\States$ for the constrained setting from admissible confidence regions $\CR\subset\HermTrace$ of the unconstrained estimation problem?
The following Lemma answers this question with a simple geometric construction:
\begin{lemma}\label{lem:ortho.admissible_truncation}
  Let $\CR$ denote an admissible confidence region for the unconstrained estimation problem for the parameter $\varrho_0 \in \HermTrace$.
  Then, $\CR^\cap := \CR\cap\States$ is an admissible confidence region for the constrained problem with $\varrho_0 \in \States$.
\end{lemma}
\begin{proof}
  Under the assumption that $\CR^\cap$ is not admissible, there must exist a ``better'' confidence region $\CR^+$ for the constrained parameter estimation problem.
  W.l.o.g.\ assume that both $\CR^+$ and $\CR^\cap$ have the same coverage.
  Therefore, we must have $\Vol(\CR^+(\vec y)) \leq \Vol(\CR^\cap(\vec y))$ for almost all observations $y \in \BR^m$, and there is a set $Y \subset \BR^m$ of non-zero measure such that $\Vol(\CR^+(\vec y)) < \Vol(\CR^\cap(\vec y))$ for $\vec y \in Y$.
  Define a new confidence region for the unconstrained problem
  \begin{equation}
    \CR' := \CR^+ \cup \CR^c,
    \label{eq:ortho.new_region}
  \end{equation}
  where $\CR^{c}=\CR\setminus \CR^{\cap}$ denotes the compliment of $\CR^{\cap}$ in $\CR$.
  Then, $\CR'$ has the given coverage level, since $\CR^+$ provides coverage for $\varrho_0 \in \States$, whereas $\CR^c$ provides coverage for the case $\varrho_0 \in \HermTrace \setminus \States$.
  Furthermore, we have for almost all $\vec y$
  \begin{equation}
    \label{eq:ortho.volumes}
    \begin{split}
      \Vol(\CR'(\vec y))
      &= \Vol(\CR^+(\vec y)) + \Vol(\CR^c(\vec y)) \\
      &\leq \Vol(\CR^\cap(\vec y)) + \Vol(\CR^c(\vec y)) \\
      &= \Vol(\CR(\vec y)).
    \end{split}
  \end{equation}
  Finally, strict inequality holds in Eq.~\eqref{eq:ortho.volumes} for all $\vec y \in Y$ due to the assumption on $\CR^+$.
  However, this would imply $\CR$ not being admissible in contradiction to the assumptions of the Lemma.
\end{proof}

One criticism raised against the use of the truncated confidence regions is the possibility that they may yield empty realizations and, hence, are considered ``unphysical''~\cite{Feldman_1998_Unified}.
However, according to the standard definition in Sec.~\ref{sub:intro.error_regions}, a procedure that reports 95\% confidence regions is allowed to give any result 5\% of the time.

Furthermore, a different strategy often adopted for point estimator is to use an unconstrained parametrization for the constrained parameter space.
A typical example is a coin toss model with bias $p \in [0, 1]$.
Instead of $p$, the problem can also be parameterized in terms of of log-odds $\log\frac{p}{1 - p}$, which can take any value in $(-\infty,\infty)$.
Similar, one could use the following parametrization for quantum states guaranteed to give a positive semidefinite, Hermitian matrix with trace 1
\[
  \rho(X) = \frac{X \adj X}{\Tr X \adj X}
\]
with $X \in \mathbb{C}^{d \times d}$.
Although this parametrization can certainly be advantageous for point estimation, it is unlikely to be helpful for uncertainty quantification:
While $X$ and $\rho(X)$ carry equivalent information, the size of a region measured in ``$X$-space'' is hardly related to the size of a region in the physical state space.
This is necessarily so, as any map from an unbounded space onto the compact quantum state space must grossly distort the geometry.
So, having obtained a ``small confidence region'' in parameter space does not imply that the state has been well-estimated w.r.t.\ any physically relevant metric.

\subsection{Confidence Regions from Linear Inversion}
\label{sub:ortho.linear_inversion}

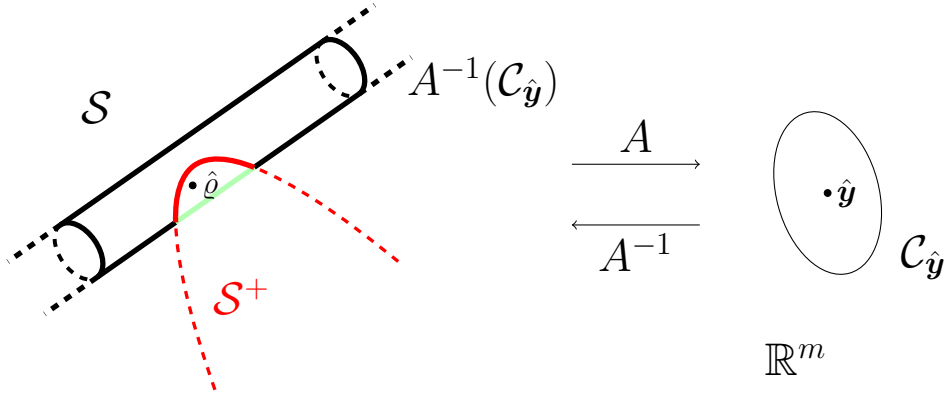
\begin{figure*}
  \centering
  \begin{tikzpicture}[scale=0.85]
    \begin{scope}[shift={(0,0)},rotate=35]
      \draw[line width=2](0,0)--(5,0)  (0,-1)--(1.5,-1)  (3,-1)--(5,-1);
      \draw[dashed, line width=2](0,0)--(-1,0)  (0,-1)--(-1,-1)  (5,0)--(6,0)  (5,-1)--(6,-1);
      \draw[color=green!30!white,line width=2](1.5,-1)--(3,-1);
      \draw[dashed, line width=1.5](0,0) arc (90:270:.3cm and .5cm);
      \draw[dashed, line width=1.5](5,0) arc (90:270:.3cm and .5cm);
      \draw[line width=2](0,0) arc (90:-90:.3cm and .5cm);
      \draw[line width=2](5,0) arc (90:-90:.3cm and .5cm);
      \draw[color=red, line width =2](1.5,-1) parabola bend (2.25,-7/16) (3,-1) ;
      \draw[color=red, dashed, line width =1.25](.5,-3.5) parabola[bend at end] (2.25,-7/16);
      \draw[color=red, dashed, line width =1.25](4,-3.5) parabola[bend at end] (2.25,-7/16);
      \draw[fill=black] (2.05,-.68) circle (1.5pt)  node[anchor=west]{{\large $\estim{{\varrho}}$}};
      \draw[color=red] (1.2,-2.2)  node[anchor=west]{{\LARGE $\States$}};
      \draw (1.2,1.4)  node[anchor=west]{{\LARGE $\HermTrace$}};
      \draw (5.6,-1.2)  node[anchor=west]{{\LARGE $A^{-1}(\CR_{\estim{\vec y}})$}};
    \end{scope}
    \begin{scope}[shift={(7,-2.5)}]
      \draw (0,0) (2,3.5) node[anchor=south]{{\LARGE $A$}} (2,2.5) node[anchor=north]{{\LARGE $A^{-1}$}};
      \draw[->] (1,3.45) -- (3,3.45);
      \draw[<-] (1,2.5) -- (3,2.5);
    \end{scope}
    \begin{scope}[shift={(10,-2.5)}]
      \draw[fill=black] (0,0) (2,3) circle (1.5pt) node[anchor=west]{{\large $\estim{\vec{y}}$}};
      \draw[rotate around={15:(2,3)}] (0,0) (2,3) ellipse (.8 and 1.3);
      \draw (0,0) (3.5,1.5) node[anchor=south]{{\LARGE $\CR_{\estim{\vec y}}$}};
      \draw (0,0) (1.5,0) node[anchor=south]{{\LARGE $\mathbb{R}^m$}};
    \end{scope}
  \end{tikzpicture}
  \caption{\label{fig:ortho.geometry}
    Geometric construction of confidence region for $\estim\varrho$.
    Quantum states are mapped by a measurement matrix $A$ to the respective quantum expectation values $\vec y$.
    Conversely, the pre-image of a confidence region $\CR_{\estim{\vec y}}$ under $A$ gives rise to a confidence region for $\estim\varrho$.
    These may be unbounded if the measurements are not tomographically complete -- a drawback that can be cured by taking into account the physical constraints on quantum states, i.e.\ positivity.
    }
\end{figure*}

A particularly simple method to transform estimates of measurement data to estimates of quantum states is the method of linear inversion, which we are going to review now:
First, assume that the \emph{true} but unknown quantum state is represented by a $d\times d$ density matrix $\varrho_0$ and the QST performed by measuring $m\geq d^{2} - 1$ tomographically-complete measurement projectors $E_{1},\ldots,E_{m}$.
By $y_{k}=\tr\left(E_{k}\varrho_0\right)$, $k=1,\ldots,m$ we denote the (quantum) expectation values of $E_{k}$ for the true state $\varrho_0$.
Since these relations are linear, we can rewrite them as $\vec{y} = A \vec{\varrho}$, where $\vec{\varrho}$ stands for the quantum state interpreted as a vector and $A$ is the measurement (or design) matrix independent of $\varrho$.
The desired (pseudo)-inverse of the above relation is
\begin{equation}
  \label{eq:ortho.linear_inversion}
  \vec{\varrho}={\left(A^{T}A\right)}^{-1}A^{T}\vec{y}
\end{equation}
and simplifies to $\vec{\varrho}=A^{-1}\vec{y}$ if $m=d^{2} - 1$.

Of course, in an experiment, the expectation values $\vec{y}$ are unknown and can only be approximated by some estimate $\estim{\vec y}$ based on the observed data.
The linear inversion estimate for the quantum state $\estim{\varrho}$ is then given by Eq.~\eqref{eq:ortho.linear_inversion} with the probabilities $\vec y$ replaced by the empirical frequencies $\estim{\vec y}$.
However, due to statistical fluctuations the estimated state $\estim{\varrho}$ is not necessarily positive semidefinite~\cite{Knips_2015_How}, which led to the development of estimators enforcing the physical constraints such as the maximum likelihood estimator~\cite{Hradil_2004_3}.
Although the linear inversion and maximum likelihood estimator solve two distinct problems -- namely the unconstrained and constrained one, respectively -- in certain cases the two are related.
 More precisely, if the outcomes approximately follow a Gaussian distribution, a fast projection algorithm computes the maximum likelihood estimate from the linear inversion estimate directly~\cite{Smolin_2012_Maximum}.

Here, we take a similar approach.
First, the simple geometric interpretation of the linear inversion estimator (see Fig.~\ref{fig:ortho.geometry}) allows us to map confidence regions for the expectation values to confidence regions for the state without taking into account the positivity constraint:
If $\CR_{\estim{\vec y}}$ is a confidence region for $\estim{\vec y}$ with confidence level $\alpha$, then so is its pre-image under the measurement map
\begin{equation}
  \label{eq:ortho.linear_confidence}
  \CR_{\estim{\vec \varrho}} := A^{-1}(\CR_{\estim{\vec y}})
\end{equation}
for $\estim{\vec \varrho}$.
Second, the truncation construction from Sec.~\ref{sub:ortho.optimal} yields an improved confidence region for the problem with quantum constraints taken into account.
As shown in Lemma~\ref{lem:ortho.admissible_truncation}, this approach yields admissible confidence region provided the original region was admissible.

The same construction can also be carried out for tomographically incomplete measurements, i.e.\ for $m<d^{2}$:
Since the measurement matrix $A$ is non-invertible in this case, the estimate for the state $\estim{\varrho}$ satisfying $A\estim{\vec\varrho} = \estim{\vec y}$ is not uniquely defined.
However, under additional structural assumptions, one can single out a unique estimate~\cite{Gross_2010_Quantum,Flammia_2012_Quantum}.
The singularity of the measurement map $A$ also reflects in the confidence region defined by Eq.~\eqref{eq:ortho.linear_confidence}.
Even if $\CR_{\estim{\vec y}}$ is a bounded region, the confidence region for the state $\CR_{\estim{\vec \varrho}}$ extends to infinity in the directions ``unobserved by $A$''.
In both cases, the tomographically complete and incomplete, we can use the intersection with the psd states to reduce the the region's size while not sacrificing coverage.
This improvement is especially far-reaching in the latter case, where it turns an unbounded region to a bounded one just by taking into account the physical constraints.\\

Of course, the question is whether we can somehow characterize the truncated confidence region $\CR_{\estim{\vec \varrho}}^\cap := A^{-1}(\CR_{\estim{\vec y}}) \cap \States$ computationally efficiently.
Since our goal is to show that this is an intractable problem exclusively due to the quantum constraints -- and not because of the complexity of high-dimensional statistics in general -- we are going to make the simplifying assumption that the measured frequencies are approximately Gaussian distributed.
Furthermore, we are going to focus on a class of confidence regions that are efficiently characterizable in the unconstrained setting, namely Gaussian confidence ellipsoids or, more precisely, ellipsoidal balls of the form
\begin{equation}
  \CR_{\estim{\vec y}} = \left\{ \vec y \in \BR^m\colon  {\left(\vec{y}-\estim{\vec{y}}\right)}^T B\left(\vec{y}-\estim{\vec{y}}\right) \leq 1  \right\}
  \label{eq:ortho.confelip}
\end{equation}
centered at the the empirical frequencies  $\estim{\vec y}$.
The $m\times m$, symmetric, positive semidefinite matrix $B$ completely specifies the ellipsoidal shape of this confidence region.
These are the natural generalizations of the well-known $2\sigma$ confidence intervals to multivariate Gaussian distributions.

However, in the unconstrained setting, the ellipsoidal construction~\eqref{eq:ortho.confelip} is known to be admissible only for $m=\left\{ 1,2\right\}$~\cite{Joshi_1969_Admissibility}, while it is not admissible for $m\geq3$~\cite{Joshi_1967_Inadmissibility} due to Stein's phenomenon~\cite{Stein_1956_Inadmissibility}.
Smaller confidence ellipsoids with the same coverage can be obtained by shifting the center slightly~\cite{Tseng_1997_Good,Hwang_1982_Minimax}.
Furthermore, other constructions similar to an egg~\cite{Shinozaki_????_Improved} or the non-convex Pascal lima\c{c}on~\cite{Brown_1995_Optimal} are known to outperform the standard ellipsoids.
Nevertheless, non of these constructions is known to be optimal and, to the best of the author's knowledge, no optimal confidence region for multivariate Gaussians in higher dimensions is known.

But, since our discussion is focused on the question how the physical psd constraints can be used to improve confidence regions, we are still going to use the ellipsoids~\eqref{eq:ortho.confelip} as a tractable example:
As we will prove later, it is impossible to characterize the truncated ellipsoids efficiently although they are fully described by a few parameters, namely $\estim{\vec y}$ and $B$ in the unconstrained case.
In other words, we show that there exists a trade-off between computational and statistical efficiency for the problem of determining ``good'' confidence regions in QST.\\

In the remainder of this section, we are going to discuss a useful parametrization for the aforementioned ellipsoids~\eqref{eq:ortho.confelip}.
To this end we use the fact that any $d\times d$ Hermitian matrix can be expanded in a basis formed by the identity $\1$ and $d^{2}-1$ traceless Hermitian matrices $\sigma_{i}$, $i=1,\ldots,d^{2}-1$, normalized according to $\textrm{Tr}(\sigma_{i}\sigma_{j})=2\delta_{ij}$.
With the symbols $\sigma_{i}$ we associate here the most common choice of the basis elements~\cite{Kimura_2003_Bloch}  --  explicitly provided in~\ref{sec:parametrisation} -- while any other $\sigma_{i}'=\sum_{j}O_{ji}\sigma_{j}$, given in terms of an orthogonal $d^{2}-1$ dimensional matrix $O$, are valid alternatives.
For $d=2$ the choice stated in~\ref{sec:parametrisation} is simply the Bloch basis of Pauli matrices: $\sigma_{1}\equiv\sigma_{x}$, $\sigma_{2}\equiv\sigma_{y}$ and $\sigma_{3}\equiv\sigma_{z}$.
In higher dimensions the matrices $\sigma_{i}$ maintain the Bloch basis structure:
Let
\begin{equation}
  \label{eq:ortho.x_yz_index}
  i_{d}=d(d-1)/2,
\end{equation}
then their construction mimics $\sigma_{x}$ for $1\leq i\leq i_{d}$, $\sigma_{y}$ for $i_{d}+1\leq i\leq2i_{d}$ and $\sigma_{z}$ for $2i_{d}+1\leq i\leq d^{2}-1$.
Therefore, we are going to refer to the $\sigma_{i}$ as (generalized) Bloch representation.

We are in position to provide the first result of this paper, falling into the category of geometry of quantum states:
\begin{theorem}\label{thm:ortho.ellipsoids}
  For the tomographically complete case $m \geq d^2 - 1$, the pre-image under the measurement matrix of any confidence ellipsoid of the form~\eqref{eq:ortho.confelip} can be represented as
  \begin{equation}
    \CR = \left\{ \estim{\varrho} + \sum_{i}R_{i}u_{i}\sigma_{i}'\colon \vec{u}^{T}\vec{u}\leq1 \right\},
    \label{eq:ortho.ellipsoid}
  \end{equation}
  where $\estim{\varrho}$ is the linear inversion estimator, that is a Hermitian matrix with $\Tr \estim{\varrho} = 1$, and the $R_{i}>0$ ($i=1,\ldots,d^{2}-1$) are the ellipsoid's radii in the directions given by $\sigma_{i}'=\sum_{j}O_{ji}\sigma_{j}$, respectively.
  The orthogonal matrix $O\in\mathcal{O}\left(d^{2}-1\right)$ furnishes any orientation of the semi-major axes of the ellipsoid.
\end{theorem}
\begin{proof}
  Note that whenever the sum has no limits specified (like in Eq.~\eqref{eq:ortho.ellipsoid}), by default it extends from $1$ to $d^{2}-1$.
  Let us parameterize both $\varrho \in \CR$ and $\estim{\varrho}$ in the Bloch representation with coordinates $w_{i}$ and $\estim{w}_i$, respectively:
  \begin{equation}
    \varrho=\frac{\1}{d}+\sum_{i}w_{i}\sigma_{i},\qquad\estim{\varrho}=\frac{\1}{d}+\sum_{i}\estim{w}_{i}\sigma_{i}.
  \end{equation}
  Since $\vec{y}=\textrm{Tr}\left(\vec{E}\varrho\right)$, and $\estim{\vec{y}}=\textrm{Tr}\left(\vec{E}\estim{\varrho}\right)$ we find
  \begin{equation}
    \vec{y}-\estim{\vec{y}}=Q\left(\vec{w}-\estim{\vec{w}}\right),
  \end{equation}
  where $Q$ is a $m\times(d^{2}-1)$ matrix with elements $Q_{ki}=\textrm{Tr}\left(E_{k}\sigma_{i}\right)$.
  In other words, the Bloch coordinates satisfy the same ellipsoid equation (\ref{eq:ortho.confelip}) as the measurement outcomes with $B$ substituted by the $d^{2}-1$ dimensional square matrix $B'=Q^{T}BQ$.
  Since $B$ is symmetric and positive definite, the same holds for $B'$.
  Hence, $B'$ can be diagonalized to the form $B'=ODO^{T}$, where $O$ is some orthogonal $d^{2}-1$ dimensional matrix and $D=\textrm{diag}(R_{1}^{-2},\ldots,R_{d^{2}-1}^{-2})$ is the diagonal matrix with positive entries.
  If we rescale $\vec{w}-\estim{\vec{w}}=OD^{-1/2}\vec{u}$, then $\vec{u}^{T}\vec{u}\leq1$ and
  \begin{equation}
    \varrho-\estim{\varrho}=\sum_{j}\left(\sum_{i}O_{ji}R_{i}u_{i}\right)\sigma_{j}.
  \end{equation}
  In the last step of the proof we simply change the orientation of the basis to $\sigma_{i}'=\sum_{j}O_{ji}\sigma_{j}$.
\end{proof}

\subsection{Computational Intractability of Truncated Ellipsoids}
\label{sub:ortho.hard}

Guided by the discussion from the previous section we now study the confidence region for the linear inversion QST defined as
\begin{equation}
  \label{eq:ortho.truncated_ellipsoid}
  \CR^\cap := \CR \cap \States = A^{-1}(\CR_{\estim y}) \cap \States,
\end{equation}
where $\CR$ is given by the ellipsoid~\eqref{eq:ortho.ellipsoid} for the tomographically complete case $m=d^2 - 1$.
In this section, we are going to show that in contrast to the full ellipsoid $\CR$, the truncated ellipsoid $\CR^\cap$ cannot be characterized computationally efficiently.
This shows, for example, in the fact that there is no efficient algorithm to answer the following question:
How much does taking into account the physical constraints reduce the size of the confidence region on a particular set of observed data?
Note that we will not be concerned with properties of the region estimator but with a single instance corresponding to a fixed set of data.
By abuse of notation, we are going to refer to these instances as $\CR$ and $\CR^\cap$ as well.

More precisely, we are concerned with the question if for a fixed ellipsoid $\CR$ there is any reduction in size due to constraints in Eq.~\eqref{eq:ortho.truncated_ellipsoid} or if $\CR$ is fully contained in the set of psd states.
For the precise formulation, we use the representation of ellipsoids from Thm.~\ref{thm:ortho.ellipsoids}.
\begin{problem}\label{prob:ortho.ellpos}
  Given the center $\estim\varrho$, radii $R_i$, and a basis $\sigma'_i$ for $\HermTrace$.
  Is there a $\vec u \in \BR^{d^2 - 1}$ with $\vec{u}^{T}\vec{u} \leq 1$ such that
  \begin{equation}
    \estim{\varrho} + \sum_{i}R_{i}u_{i}\sigma_{i}' \in \HermTrace \setminus \States?
  \end{equation}
\end{problem}
The main result of this section is the following statement on the computational complexity of the aforementioned problem.
\begin{theorem}\label{thm:ortho.hard}
  Problem~\ref{prob:ortho.ellpos} is NP-hard.
\end{theorem}
As a consequence of Thm.~\ref{thm:ortho.hard}, the problem of ``characterizing'' the truncated confidence ellipsoids $\CR_{\estim{\vec \varrho}}^\cap := A^{-1}(\CR_{\estim{\vec y}}) \cap \States$ defined in Sec.~\ref{sub:ortho.linear_inversion} computationally is hard in general.
By ``characterizing'' we mean computing any property of $\CR_{\estim{\vec \varrho}}^\cap$ that is sensitive to whether the truncation influences the original ellipsoid or not, e.g.\ computing the volume of the truncated ellipsoid or its distance to boundary of the quantum state space with high enough precision.
Note, however, that there are also properties such as the diameter that might be unaffected by the truncation in certain special cases and, hence, their computational complexity cannot be classified using Thm.~\ref{thm:ortho.hard}.
Therefore, the more general problem of computing truncated confidence regions (without the Gaussian approximation) is hard as well since it subsumes Prob.~\ref{prob:ortho.ellpos}.

Another consequence of the theorem concerns confidence regions for the constrained problem, which output ``good regions'' for the unconstrained problem when the constraints are not active:
More precisely, it is extremely natural to use likelihood ratio-based ellipsoidal confidence regions for unconstrained Gaussian data although they cannot be optimal due to Stein's phenomenon.
So it is natural to require any quantum region estimator to behave this way in the particular case that the likelihood function is concentrated well away from the boundary of state space.
What Thm.~\ref{thm:ortho.hard} shows is that any region estimator subject to this criterion must necessarily solve NP-hard problems.\\

Finally, the remainder of this section is dedicated to give some insight to the proof of the main theorem and to discuss a tractable solvable special case.
The proof of Thm.~\ref{thm:ortho.hard} is inspired by a similar result due to Ben-Tal and Nemirovski~\cite{Tal_1998_Robust} in robust optimization theory, who showed that the following problem is NP-complete.
\begin{problem}\label{prob:ortho.bental}
  Given $k$ $d\times d$ symmetric matrices $A^{1},\ldots,A^{k}$, check whether there is a $\vec u \in \BR^k$ with $\vec{u}^{T}\vec{u} \leq 1$ such that $\sum_{i=1}^{k}u_{i}A^{i} > \1_{d}$.
\end{problem}
Although the two problems are strongly related, the intractability result~\cite{Tal_1998_Robust} cannot be applied directly to our tomography related problem due to the following crucial difference:
The proof of NP-completeness of Prob.~\ref{prob:ortho.bental} deals with the case $k=d(d-1)/2+1$ and a set of real symmetric matrices ${(A_k)}_k$, which are not necessarily pairwise orthogonal to each other~\cite[Sec.~3.4.1]{Tal_1998_Robust}.
However, in Prob.~\ref{prob:ortho.ellpos}, the $\sigma'_i$ ($i=1,\ldots,d^2 - 1)$ form an orthogonal basis of the space of complex Hermitian, traceless matrices.
Hence, we need to adapt the original proof strategy to deal with the restrictions imposed by our tomography related problem.

Let us start the outline of the proof of Thm.~\ref{thm:ortho.hard} with a simplified example, when the ellipsoid in question is a ball, i.e.\ when $R_{i}=R$ for all $i=1,\ldots,d^{2}-1$.
With no loss of generality, we can assume $\sigma'_i = \sigma_i$.
The following Lemma, which is proven in~\ref{sec:spheres}, provides an easily checkable, necessary, and sufficient condition to decide Prob~\ref{prob:ortho.ellpos} for this special case.
\begin{lemma}\label{lem:ortho.spheres}
  Let $\CR$ denote a ball parameterized according to Thm.~\ref{thm:ortho.ellipsoids} with with radii $R_i=R$ and midpoint $\estim{\varrho}$.
  $\CR$ is fully contained in the set of psd density matrices if and only if
  \begin{equation}
    R\leq\sqrt{\frac{d}{2\left(d-1\right)}} \, \mineig \estim{\varrho},
  \end{equation}
  where $\mineig\estim{\varrho}$ denotes the smallest eigenvalue of $\estim{\varrho}$.
\end{lemma}
The statement is a straightforward but interesting extension of the known result that the largest ball centered at the completely mixed state and fully contained in the set of psd density matrices has radius $R_{\mathrm{max}}=\sqrt{\frac{1}{2d\left(d-1\right)}}$.
Intuitively, when the center of the ball is moved away from the completely mixed state, the allowed radii become smaller.
This correction happens to be quantified by the smallest eigenvalue of the new center.
In conclusions, spherical ellipsoids do not constitute hard instances of Problem~\ref{prob:ortho.ellpos} provided that the minimal eigenvalue of $\estim\varrho$ can be computed efficiently with high enough accuracy.

However, it turns out that a slightly more complicated setting is already enough to proof the computational intractability.
The ellipsoids under consideration still have their semi-major axes aligned with the generalized Bloch basis, that is we assume $\sigma'_i = \sigma_i$.
The only change compared to the previous setting is the choice of radii.
We consider the same radius $R_{1}$ for all directions generalizing the $x$-direction to higher dimensions and the distinct radius $R_{2}$ for the remaining directions:
\begin{equation}
  \label{eq:ortho.subclass}
  \begin{split}
    R_{i}=R_{1} &\quad i=1,\ldots,i_{d}\\
    R_{i}=R_{2} &\quad i=i_{d}+1,\ldots,d^{2}-1.
  \end{split}
\end{equation}
Recall $i_d$ defined in Eq.~\eqref{eq:ortho.x_yz_index}.

Now, in order to prove the computational intractability of Problem~\ref{prob:ortho.ellpos}, we use a reduction from the \emph{balanced sum} problem, which is known to be NP-complete.
\begin{problem}\label{prob:ellpos.balanced_sum}
  Given a vector $\vec a \in \mathbb{N}^d$, decide whether there exists a vector $\vec \psi$ with
  \begin{equation}
    \forall{k}\:\psi_{k}\in\left\{ -1,1\right\} \;\textrm{ and }\quad \vec a \cdot \vec\psi=0.
    \label{eq:ellpos.partition_vector}
  \end{equation}
\end{problem}
In case there is such a vector $\vec \psi$ one says that the instance $\vec a$ allows for a balanced sum partition because the sum of components of $\vec a$ labeled by $\psi_i = 1$ is equal to the sum of components $a_i$ labeled by $\psi_i = -1$.
The main technical difficulty is now to identify the values of $R_1$ and $R_2$ as well as $\estim{\rho}$ depending on an instance of the balanced sum problem $\vec a$ such that the corresponding ellipsoid $\CR$ given by Thm.~\ref{thm:ortho.ellipsoids} contains an element with negative eigenvalues if and only if $\vec a$ has a balanced sum partition.
For the details, please see~\ref{sec:ellpos}.

\section{Bayesian Credibility regions}
\label{sec:bayesian}

\subsection{MVCR for Gaussians}
\label{sub:bayesian.gaussian}

We now turn to the question of minimal volume credible regions (MVCR) in the Bayesian framework:
In the unconstrained case, Gaussian posteriors are one of the few examples of multivariate distributions, where the MVCR are simple geometric objects, namely ellipsoids.
In practice, Gaussian posteriors arise in the following scenario:
Consider a random vector $\vec X\sim\mathcal{N}(\vec\mu, \Sigma)$, where the covariance matrix $\Sigma$ is known and we wish to estimate its mean $\vec\mu$.
If we furthermore assume a Gaussian prior for the mean, the posterior will be Gaussian as well due to the fact that the Gaussian distribution is its own conjugated prior.

This is one of the few cases, in which the Bayesian update as well as the computation of an optimal credible region can be carried out analytically.
First, computing the parameters for the Gaussian posterior distribution can be done by means of \emph{linear Kalman filter update equations}, see e.g.~\cite[Sec.\ 2.4]{Audenaert_2009_Quantum}.
Second, for the credible region, assume that after the Bayesian update, the posterior distribution of $\vec\mu$ is parameterized by its mean $\vec\theta \in \BR^N$ and covariance matrix $\Sigma \in \BR^{N \times N}$.
Therefore, the posterior of $\vec\mu$ has probability density
\begin{equation}
  \label{eq:bayesian.gaussian_density}
  \pi_{\vec\theta,\Sigma}(\vec x) = {(2\pi)}^{-\Nhalf} \abs{\Sigma}^{-\frac{1}{2}} \exp\left( -\frac{1}{2} \norm{\vec x - \vec\theta}_\Sigma^2 \right).
\end{equation}
where
\begin{equation}
  \norm{\vec x - \vec\theta}_\Sigma := \sqrt{{(\vec x - \vec\theta)}^T \Sigma^{-1} (\vec x - \vec\theta)}
\end{equation}
is the Mahalanobis distance and  $\abs{\Sigma}$ denotes the determinant of $\Sigma$.
As elaborated in Sec.~\ref{sub:intro.error_regions}, the MVCRs are exactly the highest posterior density sets as defined in Eq.~\eqref{eq:intro.highest_posterior}.
Therefore, the MVCR with credibility $\alpha$ for the density Gaussian~\eqref{eq:bayesian.gaussian_density} is given by
\begin{equation}
  \label{eq:bayesian.gaussian_cr}
  \CR = \{ \vec x \in \BR^N\colon \norm{\vec x - \vec\theta}_\Sigma \le r_{\alpha} \} =: \Eps(r_{\alpha}).
\end{equation}
This is an ellipsoid centered at $\vec\theta$ with radius $r_{\alpha}$ determined by the saturated credibility condition~\eqref{eq:intro.credible_region}:
\begin{equation}
  \label{eq:bayesian.radius}
  \begin{split}
    \alpha
    &= {(2\pi)}^{-\Nhalf} \abs{\Sigma}^{-\frac{1}{2}} \int_{\Eps(r_{\alpha})} \exp\left( -\frac{1}{2} \norm{\vec x - \vec\theta}_\Sigma^2 \right) \dd^N x \\
    &= \frac{\gamma\left( \Nhalf, \frac{r^2_{\alpha}}{2} \right)}{\Gamma\left( \Nhalf \right)}
    \equiv P\left( \Nhalf, \tfrac{r^2_{\alpha}}{2} \right).
    \end{split}
\end{equation}
By $\gamma(\cdot,\cdot)$ we denote the incomplete $\Gamma$-function and $P(\cdot,\cdot)$ is its normalized version.
The above condition fixes $r_{\alpha}$ uniquely since $x \mapsto P(\Nhalf, x)$ is strictly monotonic for any $N > 0$.
Hence, determining the MVCR for a multivariate Gaussian posterior with known mean and covariances reduces to computing the radius $r_\alpha$, which is formalized in the following problem.
\begin{problem}\label{prob:bayesian.cr}
  For given mean $\vec\theta \in \BR^N$, covariance matrix $\Sigma \in \BR^{N \times N}$ with $\Sigma \geq 0$, credibility $\alpha \in [0,1]$, and accuracy $\delta$ with $\delta^{-1} \in \mathbb{N}$, determine the radius of the MVCR $r_{\alpha}$ defined in Eq.~\eqref{eq:bayesian.radius} with given accuracy.
\end{problem}

An efficient algorithm for solving Prob.~\ref{prob:bayesian.cr} is outlined in the following.
To ease notation, we set $x = r^2_{\alpha}/2$.
\begin{enumerate}
  \item W.l.o.g.\ we can assume that $\alpha \le 0.9$ (or some other arbitrary constant).
  Otherwise, the problem can be restated in terms of $Q(\Nhalf, x) = 1 - P(\Nhalf, x)$, which allows for a similar analysis.
  The condition $\alpha \le 0.9$ restricts the search space for $x$ to some finite interval $[0, t_\mathrm{max}]$.
  Note that the upper bound $t_\mathrm{max}$ grows at worst polynomially in $\Nhalf$.
  \item The above restriction, the finite precision, and the fact that $x \mapsto P(\Nhalf, x)$ is strictly monotonic allow for interpreting the problem of finding $x$ given $\alpha$ as a search in an ordered, finite list of size $M \sim \tfrac{t_\mathrm{max}}{\delta}$.
  \item Each entry of this list can be evaluated with exponential precision in polynomial time using a power series expansion of $P(\Nhalf, x)$ (for more details see Lemma~\ref{lem:bayesian.normalization_constant} in~\ref{sec:proof_bayesian}).
  \item Since finding $x$ in this list only requires $\log M$ evaluations using binary search, the whole problem can be solved in polynomial time.
\end{enumerate}

\subsection{Bayesian QST}
\label{sub:bayesian.tomography}

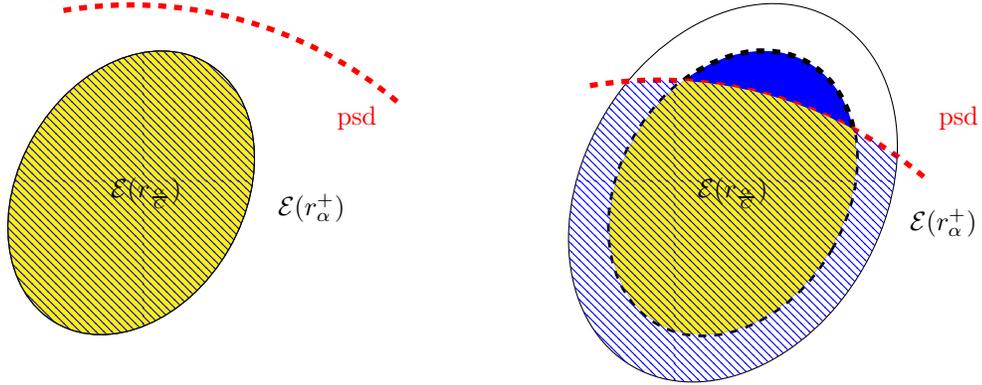
\begin{figure*}[t]
  \centering
  \begin{tikzpicture}[
    elip/.style={fill=yellow, line width=0},
    truncated/.style={pattern=north west lines, pattern color=blue},
    truncelip/.style={dashed, line width=2pt, fill=blue},
    psd/.style={dashed, color=red, line width=2pt}
  ]
    \clip (-3,-3) rectangle (13,3);

    \def\origelip{(0,0) ellipse [x radius=1.5, y radius=2, rotate=-30]}
    \def\truncatedelip{(0,0) ellipse [x radius=2, y radius=2.66666, rotate=-30]}

    \def\psdradius{5.5}
    \def\psdcircle{(psdcenter) [draw=none] circle (\psdradius)}

    \begin{scope}[shift={(0,0)}]
      \draw[elip]\origelip;
      \draw[truncated]\origelip;
      \node at (0.2,0) {$\Eps(r_{\frac{\alpha}{C}})$};
      \node at (2.4,-.2) {$\Eps(r^+_{\alpha})$};

      \node at (0,-3) (psdcenter) {};
      \node at (3,1) {\red{psd}};
      \draw[psd] ([shift=(50:\psdradius)]psdcenter) arc (50:100:\psdradius);
    \end{scope}

    \begin{scope}[shift={(8,0)}]
      \draw[truncelip]\origelip;
      \draw\truncatedelip;

      \node at (-1,-4) (psdcenter) {};
      \node at (3,1) {\red{psd}};

      \draw[psd] ([shift=(50:\psdradius)]psdcenter) arc (50:100:\psdradius);
      \begin{scope}
        \clip \psdcircle;
        \fill[elip] \origelip;
        \fill[truncated]\truncatedelip;
      \end{scope}

      \node at (0,0) {$\Eps(r_{\frac{\alpha}{C}})$};
      \node at (2.8,-.4) {$\Eps(r^+_{\alpha})$};
    \end{scope}
  \end{tikzpicture}
  \caption{\label{fig:bayesian.ellipsoids}
    The two possible cases for the credible regions.
    \emph{Left:} The original ellipsoid $\Eps(r_\frac{\alpha}{C})$ with credibility $\frac{\alpha}{C}$ (yellow) lies completely inside the psd states and is, therefore, equal to the ellipsoid taking into account positivity $\Eps(r^+_{\alpha})$ with credibility $\alpha$ (blue hatched).
    \emph{Right:}  Parts of the original ellipsoid $\Eps(r_\frac{\alpha}{C})$ lie outside the psd states (blue).
    Hence, the ellipsoid that takes into account positivity $\Eps(r^+_{\alpha})$ has to have a larger radius in order to achieve the sought for credibility.
  }
\end{figure*}

Let us now turn to the application of Bayesian methods to QST, for a more thorough discussion see e.g.~\cite{Granade_2015_Practical}.
In order to incorporate the prior knowledge of positive semidefiniteness, we chose a prior that is concentrated on $\States$ and vanishes on its complement.
As before% in Sec.~\ref{sub:bayesian.gaussian}
, we choose a (truncated) Gaussian prior and, therefore, Gaussian posteriors.
Hence, the density $\pi_{\theta,\Sigma}^+(\varrho)$ of a Gaussian posterior on $\States$ with respect to the flat Hilbert-Schmidt measure $\dd\varrho$ on $\HermTrace$ can be written as
\begin{equation}
  \label{eq:bayesian.density_plus}
  \pi^+_{\theta,\Sigma}(\varrho) = C_{\theta,\Sigma}\ \chi(\varrho)\ \pi_{\theta,\Sigma}(\varrho).
\end{equation}
Here, $\pi_{\theta,\Sigma}$ is the multivariate Gaussian from Eq.~\eqref{eq:bayesian.gaussian_density} with $\theta \in \HermTrace$.
The other factors in Eq.~\eqref{eq:bayesian.density_plus} ensure that $\pi^+_{\theta,\Sigma}$ is a proper probability distribution supported on $\States$:
$\chi(\varrho)$ is the indicator function of $\States$ and $C_{\theta,\Sigma}$ is the normalization constant defined by
\begin{equation}
  C_{\theta,\Sigma}^{-1} = \int_{\States} \pi_{\theta,\Sigma}(\varrho) \,\dd\varrho.
\end{equation}
From now on we will drop the subscripts indicating the mean $\theta$ and the covariance matrix $\Sigma$ if no confusion arises. It is then important to remember that the constant in question is denoted by $C$, while the credibility region is $\CR$.

The problem we try to solve is the following:
Given the mean $\theta$, covariance matrix $\Sigma$, and credibility $\alpha$, can we find the MVCR for the Gaussian distribution supported on $\States$?
Since the posterior density~\eqref{eq:bayesian.density_plus} is supported on the psd states and MVCRs are highest-density sets due to~\eqref{eq:intro.highest_posterior}, the MVCR is of the form
\begin{equation}
  \Eps(r^+_{\alpha}) \cap \States = \{ \varrho \in \States \colon \norm{\varrho - \theta}_\Sigma \le r^+_{\alpha} \}.
  \label{eq:bayesian.mvcr}
\end{equation}
Similar to Eq.~\eqref{eq:bayesian.radius}, the radius is determined by the credibility condition
\begin{equation}
  \label{eq:bayesian.radius_plus}
  \alpha = C \int_{\Eps(r^+_{\alpha}) \cap \States} \pi_{\theta,\Sigma}(\varrho) \dd \varrho.
\end{equation}
However, this case involves the normalization constant $C$ from~\eqref{eq:bayesian.density_plus} and the integral is restricted to the psd states.
Also, there is no closed-form analogue to Eq.~\eqref{eq:bayesian.radius} due to the psd constraint.

\subsection{Computational Intractability}
\label{sub:bayesian.hard}

Our main result from this section concerns MVCR for Gaussian posteriors that are fully supported on the psd states.
We will show that the following problem is computational hard.
\begin{problem}\label{prob:bayesian.trucated_cr}
  For given mean $\theta \in \HermTrace$, covariance matrix $\Sigma$, credibility $\alpha \in [0,1]$, and accuracy $\delta$ with $\delta^{-1} \in \mathbb{N}$, determine the radius of the MVCR $r^+_{\alpha}$ defined in Eq.~\eqref{eq:bayesian.radius_plus} with given accuracy.
\end{problem}
In other words, there is no efficient algorithm that outputs smallest volume credibility regions for every Gaussian distribution on $\HermTrace$ restricted to the positive semidefinite states and every credibility $\alpha$.
Consequently, there cannot be an efficient algorithm to solve the problem of MVCR for QST, since the latter more general problem contains the instances of Prob.~\ref{prob:bayesian.trucated_cr}.
To prove Prob.~\ref{prob:bayesian.trucated_cr}, we use a reduction from Problem~\ref{prob:ortho.ellpos}, which has already been shown to be NP-hard.
This reduction runs along the following lines:
\begin{enumerate}
  \item Assume that Prob.~\ref{prob:bayesian.trucated_cr} can be solved efficiently.

  \item As we will prove later, every ellipsoid $\Eps^*$ in $\HermTrace$ can be encoded as a minimum volume credible ellipsoid for some Gaussian distribution $\pi$ with a suitable choice of $\theta$, $\Sigma$, and $R$:
  \begin{equation}
    \Eps^* = \Eps_{\theta,\Sigma}(R).
  \end{equation}
  Note that only $\theta$ is uniquely defined.
  $\Sigma$ is defined only up to a multiplicative, positive constant, since every rescaling of $\Sigma$ can be compensated by an appropriate rescaling of $R$.

  \item  Using the assumed efficient algorithm for Prob.~\ref{prob:bayesian.trucated_cr}, we can compute the normalization constant $C$ of the truncated distribution~\eqref{eq:bayesian.density_plus} for given $\theta$ and $\Sigma$ with sufficient precision in polynomial time.

  \item Based on this, we can compute a credibility $\alpha$ such that $R = r_\frac{\alpha}{C}$ and, therefore,
  \begin{equation}
    \Eps^* = \Eps_{\theta,\Sigma}(r_\frac{\alpha}{C}).
  \end{equation}

  \item The crucial observation is that this ellipsoid is contained in the psd states if and only if the corresponding MVCR for the truncated distribution $\pi^+$ fulfills
  \begin{equation}
    \label{eq:bayesian.criterion}
    r^+_{\alpha} = r_\frac{\alpha}{C}.
  \end{equation}
  See Fig.~\ref{fig:bayesian.ellipsoids} for an illustration.
  Since we can compute $r^+_{\alpha}$ efficiently by assumption, checking Eq.~\eqref{eq:bayesian.criterion} allows us to decide Prob.~\ref{prob:ortho.ellpos}.
\end{enumerate}

In conclusion, the main result from this section is the following lower bound on the computational complexity of Problem~\ref{prob:bayesian.cr}.
\begin{theorem}\label{thm:bayesian.hardness}
  If Problem~\ref{prob:bayesian.trucated_cr} has a polynomial time algorithm, then we can also decide Problem~\ref{prob:ortho.ellpos} in polynomial time.
  Therefore, there is no efficient algorithm for Problem~\ref{prob:bayesian.trucated_cr} unless $\mathrm{P} = \mathrm{NP}$.
\end{theorem}

The proof runs along the lines outlined above and can be found in~\ref{sec:proof_bayesian}.
Here, the main technical problem is that we are dealing with finite-precision arithmetic.

\section{Conclusion \& Outlook}
\label{sec:outlook}

The goal of this work is to provide an absolute ``upper bound'' on what we can expect from algorithms computing error regions for QST and to demonstrate that there is a trade-off between optimality and efficiency.
This paper should not be understood as providing a no-go theorem for efficient algorithms in practice since the negative result of this work does not rule out efficient algorithms for practically acceptable approximations to optimal regions.
Also, there is no indication that the various approaches used in practice give rise to regions that are far from optimal or do not have the advertised coverage.
The reason our result leaves room for feasible approaches in practice are twofold:
First, like any result showing NP-hardness, we prove that there is no efficient algorithm solving the exact problem deterministically for any instances.
Hence, our result neither precludes the existence of efficient approximate or probabilistic algorithms, nor cannot make any statement about average case hardness.
Second, although the experimental effort necessary for full-fledged tomography scales polynomially in the dimension of the system -- and is, therefore, efficient in the sense of computational complexity -- in practice other characterization techniques such as randomized benchmarking or direct fidelity estimation become more important for larger dimensions.
It should now be the goal of future work to further close down the gap between existing positive results and the proven no-go theorems from either side.

More specifically, due to the simplifying assumptions made we investigate computational intractability that is solely caused by the quantum constraints and not by the general complications in high-dimensional statistics.
In the Bayesian settings we show that minimal volume (w.r.t.\ the Hilbert-Schmidt measure) credible regions for truncated Gaussian posterior distributions are hard to compute.
Therefore, the problem of determining MVCR for QST cannot be solved efficiently as well, since any algorithm solving the latter must also be able to solve instances with the specific prior used in Prob.~\ref{prob:bayesian.trucated_cr}.

The result for frequentist confidence regions is somewhat weaker since optimal confidence regions for high-dimensional Gaussian distributions are not known for most natural notions of optimality.
Nevertheless, Gaussian confidence ellipsoids constitute a viable choice due to their simplicity and tractability.
However, our results show that the constraints imposed by quantum mechanics render the task of characterizing the confidence regions for the constrained problem computationally intractable -- even under the simplifying assumptions made.
Of course, any more general setting encompassing the Gaussian approximation will be at least as hard to treat as the one used in this work.
Furthermore, it also shows that computing any confidence region estimator yielding ellipsoids when the constraints are not active (and anything possibly better when they are) involves solving NP-hard problems.\\

Recently, the mathematical statistics community has started to analyze the trade-offs between computational complexity and optimality in inference problems -- see e.g.\
\cite{Berthet_2013_Complexity,Berthet_2013_Computational,Zhang_2014_Lower}.
Early papers concentrated on the problem of \emph{sparse principal component analysis}, which roughly asks whether the covariance matrix of a random vector possess a sparse eigenvector with large eigenvalue~\cite{Berthet_2013_Complexity,Berthet_2013_Computational,Zhang_2014_Lower}.
Later works have addressed the much better-studied problem of sparse inference~\cite{Zhang_2014_Lower}.
The main difference between these papers and the present one is that we always condition on a data set and show that certain operations for quantifying uncertainty given the data are hard.
This approach is canonical for a Bayesian analysis, but merely ``natural'' for orthodox error regions (c.f.~Sec.~\ref{sub:intro.error_regions}).
In contrast, Refs.~\cite{Berthet_2013_Complexity,Berthet_2013_Computational,Zhang_2014_Lower} analyze the ``global'' performance of orthodox estimators -- i.e.\ they do not require looking at worst-case scenarios over the data.
References~\cite{Berthet_2013_Complexity,Berthet_2013_Computational,Zhang_2014_Lower} achieve this by reducing a certain problem (``hidden clique'') -- that is conjectured to be hard in the average case -- to the sparse PCA problem; while~\cite{Zhang_2014_Lower} employs a more subtle argument involving the non-uniform complexity class $P/\mathrm{poly}$.
It would be very interesting to adapt such arguments to the problem of quantum uncertainty quantification.

Of course, from the practical point of view, ``positive'' results -- i.e.\ new algorithms to solve the problem -- would be more beneficial.
Here, recent work on sampling distributions restricted to convex bodies~\cite{Cousins_2013_Cubic,Cousins_2015_Bypassing} could be a starting point for further investigations.

Beside quantum state tomography, our results might also be relevant to problems involving psd constraints such as the estimation of covariance matrices.

\section*{Acknowledgments}
This work has been supported by the Excellence Initiative of the German Federal and State Governments (Grants ZUK 43 \& 81), the ARO under contract W911NF-14-1-0098 (Quantum Characterization, Verification, and Validation), and the DFG projects GRO 4334/1,2 (SPP1798 CoSIP).

\section*{References}

\providecommand{\newblock}{}

\appendix

\section{Generalized Bloch Representation}
\label{sec:parametrisation}

Here, we provide the particular generalizations $\sigma_{i}$ of the Pauli matrices used in Sec.~\ref{sub:ortho.linear_inversion}.
These are exactly the generators of the group $SU(d)$, see e.g.~\cite{Kimura_2003_Bloch,Byrd_2003_Characterization} for more details.
Since the exact order of the $\sigma_{i}$ is not important for our purposes, we present them as finite sets of matrices generalizing the $\sigma_\mathrm{X}$, $\sigma_\mathrm{Y}$, and $\sigma_\mathrm{Z}$ matrix, respectively:
\begin{equation}
  \left\{ \sigma_{i}:i=1,\ldots,i_{d}\right\} =\left\{ \Xi_{jk}^{(\textrm{Re})}:1\leq j<k\leq d\right\} ,
\end{equation}
\begin{equation}
  \left\{ \sigma_{i}:i=i_{d}+1,\ldots,2i_{d}\right\} =\left\{ \Xi_{jk}^{(\textrm{Im})}:1\leq j<k\leq d\right\} ,
\end{equation}
\begin{equation}
  \left\{ \sigma_{i}:i=2i_{d}+1,\ldots,d^{2}-1\right\} =\left\{ \Xi_{l}^{(\textrm{diag})}:1\leq l\leq d-1\right\} .
\end{equation}
Recall that $i_{d}=d(d-1)/2$.
The matrices on the right hand side are defined in terms of some orthonormal basis ${\{\ket{i}\}}_i$:
\begin{equation}
  \Xi_{jk}^{(\textrm{Re})} =  \ket{j}\bra{k} + \ket{k}\bra{j},
\end{equation}
\begin{equation}
  \Xi_{jk}^{(\textrm{Im})} = -\ii \Big(\ket{j}\bra{k} - \ket{k}\bra{j}\Big),
\end{equation}
\begin{equation}
  \Xi_{l}^{(\textrm{diag})} = \sqrt{\frac{2}{l\left(l+1\right)}}\left(\sum_{j=1}^{l} \ket{j}\bra{j} - l \ket{l+1}\bra{l+1} \right).
\end{equation}

\section{Proof of Theorem~\ref{thm:ortho.hard}}
\label{sec:ellpos}

We shall start the current discussion with a word of clarification concerning the dual notation already used in the definition of the Bloch vector. We utilize an alternative representation of the state $\ket{\Psi} $ in terms of a complex vector $\vec{\psi}$ with coordinates
\begin{equation}
  \psi_{k}=\left\langle k \ket{\Psi} \right.,\qquad k=1,\ldots,d,\label{coordinates}
\end{equation}
specified with respect to the orthonormal basis fixed in~\ref{sec:parametrisation}.
Consequently,  $\sqrt{\left\langle \Psi \ket{\Psi} \right.}$ is the  norm of $\ket{\Psi}$, while $\norm{\vec\psi}$ denotes the norm of $\vec\psi$. Obviously both norms assume the same value.

In a first step of the proof we write down the positivity condition for the ellipsoid under investigation:
The confidence ellipsoid $\CR$ is fully contained in the set of psd states if and only if for all $\rho \in \CR$ and all $\ket{\Psi}$,
\(
  \bra{\Psi} \rho \ket{\Psi} \ge 0.
\)
holds.
In the parametrization from Thm.~\ref{thm:ortho.ellipsoids}, this condition can be rewritten as
\begin{equation}
  \bra{\Psi}\estim{\varrho} \ket{\Psi} +R_{1}\sum_{i=1}^{i_{d}}u_{i}v_{i}\left(\vec\psi\right)+R_{2}\sum_{i=i_{d}+1}^{d^{2}-1}u_{i}v_{i}\left(\vec\psi\right)\geq0,
  \label{eq:ellpos.positivity}
\end{equation}
where we have already restricted our attention to the special case from Eq.~\eqref{eq:ortho.subclass}.
Furthermore, we have used the shorthand $v_{i}\left(\vec\psi\right)=\bra{\Psi}\sigma_{i} \ket{\Psi} $, which are the rescaled Bloch coordinates of the density matrix $\ket{\Psi} \bra{\Psi}$.
Condition~\eqref{eq:ellpos.positivity} is independent of the norm of $\ket{\Psi}$ thus, we can fix $\left\langle \Psi \ket{\Psi} \right. = d$.
Recall that Eq.~\eqref{eq:ellpos.positivity} has to hold for all values of $\vec u$ with $\vec u^T \vec u \le 1$.
Since the left hand side assumes its minimal value for
\begin{equation}
  u_{i} = -\frac{v_{i}\left(\vec\psi\right)}{\sqrt{\sum_{j}v_{i}^{2}\left(\vec\psi\right)}},
\end{equation}
we find that Eq.~\eqref{eq:ellpos.positivity} is equivalent to
\begin{equation}
 \bra{\Psi}\estim{\varrho} \ket{\Psi} -\sqrt{R_{1}^{2}\sum_{i=1}^{i_{d}}v_{i}^{2}\left(\vec\psi\right)+R_{2}^{2}\sum_{i=i_{d}+1}^{d^{2}-1}v_{i}^{2}\left(\vec\psi\right)}\geq0.
  \label{eq:ellpos.worst_case}
\end{equation}
Using the unusual normalization of $\ket\Psi$, we find
\begin{equation}
  \sum_{i}v_{i}^{2}\left(\vec\psi\right)=2 d\left(d-1\right) =: \mathcal{P},
\end{equation}
which can be utilized to simplify~\eqref{eq:ellpos.worst_case}
\begin{equation}
 g(\vec\psi) := \bra{\Psi}\estim{\varrho} \ket{\Psi} -\sqrt{\mathcal{P}R_{2}^{2}+\left(R_{1}^{2}-R_{2}^{2}\right)\sum_{i=1}^{i_{d}}v_{i}^{2}\left(\vec\psi\right)}\geq0.
  \label{eq:ellpos.major}
\end{equation}
In the following, we restrict our attention to  $R_{1}>R_{2}$, so that both term inside the square root are manifestly non-negative.\\

In the second step of the proof we show and utilize the following lemma:
\begin{lemma}\label{lem:ellpos.real_min}
  If $\estim{\varrho}$ is a symmetric, real matrix w.r.t.\ $\ket{i}$, then the minimum of $g(\vec\psi)$ is attained by a vector $\vec{\psi}$ with real coordinates.
\end{lemma}
\begin{proof}
  Note that we can decompose any vector $\ket\Psi$ into its real and imaginary part
  \begin{equation}
    \ket\Psi = \ket{\Psi_1} + \ii\Ket{\Psi_2},
  \end{equation}
  where the $\ket{\Psi_i}$ are given by real vectors $\vec\psi_i$.
  Therefore, for $\estim{\varrho}$ being real and symmetric, we find
  \begin{equation}
    \bra{\Psi} \estim\varrho \ket{\Psi} = \bra{\Psi_1} \estim\varrho \ket{\Psi_1} + \bra{\Psi_2} \estim\varrho \ket{\Psi_2}.
  \end{equation}
  A similar equality holds with $\estim\varrho$ replaced by $\1$ or $\sigma_i$ for $i=1,\ldots,i_d$, since the latter matrices are symmetric and real as well.
  To shorten the notation, we now define two $i_d+1$ dimensional vectors $\vec x^1$ and $\vec x^2$ with components ($\alpha=1,2$)
  \begin{equation}
    \begin{split}
      x^\alpha_0 &= \frac{\sqrt{\mathcal{P}}}{d} R_2 \norm{\vec\psi_\alpha}^2 \\
      x^\alpha_i &= \sqrt{R_1^2 - R_2^2}\;v_i\left(\vec\psi_\alpha\right) \qquad (i=1,\ldots,i_d).
    \end{split}
  \end{equation}
  Since $d = \norm{\vec\psi}^2 = \norm{\vec\psi_1}^2 + \norm{\vec\psi_2}^2$, we find
  \begin{equation}
  \sqrt{ \mathcal{P} R_2^2 + (R_1^2 - R_2^2) \sum_{i=1}^{i_d} v_{i}^{2}\left(\vec\psi\right)} = \norm{\vec x^1 + \vec x^2} \le \norm{\vec x^1} + \norm{\vec x^2},
  \end{equation}
  where we used triangle inequality in the last step.
  Therefore
\begin{equation}
 g(\vec\psi) \geq  g(\vec\psi_1) + g(\vec\psi_2)
\end{equation}
so that if  $g(\vec\psi)$ is non-negative for all real vectors, it is also non-negative for every complex vector $\vec\psi$. More intuitively, the above result is true because the construction of $g(\vec\psi)$ utilizes only the generalized $\sigma_x$ Pauli matrices, which by construction pick up certain real parts of $\vec\psi^*\otimes\vec\psi$ (imaginary contribution could appear only due to $\sigma_y$).
\end{proof}

The next step of the proof, which is crucial for encoding an instance  the balanced sum problem, is the choice of the ellipsoid's center.
We choose
\begin{equation}
  \estim{\varrho}=\frac{q}{d}\1+\frac{1-q}{a^{2}} \ket{\vec a} \bra{\vec a},\qquad0\leq q\leq1,\qquad a=\norm{\vec a},
  \label{eq:ellpos.rho0}
\end{equation}
with $q$ to be specified below and $\ket{\vec a} =\sum_{k}a_{k}\ket{k} $ denoting a state represented by a real, \emph{integral} vector $\vec{a}$ playing the role of the instance of Prob.~\ref{prob:ellpos.balanced_sum}. % to be encoded.
Since $\estim\varrho$ given by Eq.~\eqref{eq:ellpos.rho0} is manifestly real and symmetric we can restrict our attention to $\vec\psi \in \BR^d$ due to Lemma~\ref{lem:ellpos.real_min}. We find
\begin{equation}
  \bra{\Psi}\estim{\varrho} \ket{\Psi} =q+\frac{1-q}{a^{2}}{\left(\vec{a}\cdot\vec{\psi}\right)}^{2},
\end{equation}
and
\begin{equation}
  \sum_{i=1}^{i_{d}}v_{i}^{2}\left(\vec\psi\right)=4\sum_{1\leq j<k\leq d}\psi_{j}^{2}\psi_{k}^{2}\equiv2d^{2}-2\sum_{k=1}^{d}\psi_{k}^{4}.
\end{equation}

Before we will be ready to take an advantage of the above encoding we need to  perform a sequence of tedious algebraic manipulations. In short, the function we work with has an algebraic form $g(\vec\psi)=\kappa-\sqrt{\Delta}$, with both $\kappa$ and $\Delta$ being non-negative. Testing if this function is non-negative is thus equivalent to checking the inequality $\kappa^2- \Delta\geq 0$. If we divide this inequality by $2(R_1^2-R_2^2)$ and fix $q=q_+$ or $q=q_-$ with
\begin{equation}
  q_{\pm}=\frac{1}{2}\left(1\pm\sqrt{1-8d\left(R_{1}^{2}-R_{2}^{2}\right)\frac{a^{2}}{1+a^{2}}}\right).\label{eq:ellpos.q}
\end{equation}
we can rearrange it to the convenient form
\begin{equation}
  f\left(\vec{\psi}\right)-C_{2}{\left(\vec{a}\cdot\vec{\psi}\right)}^{4}\leq C_{1},
  \label{eq:ellpos.condition}
\end{equation}
where:
\begin{align}
  f\left(\vec{\psi}\right) &= 2d^{2}-\sum_{k=1}^{d}\psi_{k}^{4}-2d\frac{{\left(\vec{a}\cdot\vec{\psi}\right)}^{2}}{1+a^{2}}, \\
  \label{eq:ellpos.def_c1}
  C_{1}&=d^{2}+\frac{1}{R_{1}^{2}-R_{2}^{2}}\left[\frac{q_{\pm}^{2}}{2}-d\left(d-1\right)R_{2}^{2}\right], \\
  C_{2}&=\frac{q_{\mp}^{2}}{2a^{4}\left(R_{1}^{2}-R_{2}^{2}\right)}>0
\end{align}
Both solutions~\eqref{eq:ellpos.q} assure that~\eqref{eq:ellpos.condition} is free from additional terms proportional to ${\left(\vec{a}\cdot\vec{\psi}\right)}^{2}$, except those already hidden in $f$.

Hence, the original problem of deciding whether the ellipsoid $\Eps$ centered at $\estim\varrho$ and with radii~\eqref{eq:ortho.subclass} is contained in the psd states can be rephrased as deciding whether the maximum of the left hand side of Eq.~\eqref{eq:ellpos.condition} is smaller or equal to some constant:
\begin{equation}
  \Eps\subset\States
  \iff
  \max_{\vec\psi\in\mathbb{S}^{d-1}_{d}}\left[f\left(\vec\psi\right)-C_{2}{\left(\vec a \cdot \vec\psi\right)}^{4}\right]\leq C_{1}.
  \label{eq:ellpos.max_condition}
\end{equation}
Here, $ \mathbb{S}^{d-1}_{\zeta}$ denotes a $(d-1)$-dimensional sphere with radius $\sqrt\zeta$, i.e.
\begin{equation}
 \vec\psi\in\mathbb{S}^{d-1}_{d} \iff    \vec\psi\in\mathbb{R}^{d}\;\wedge \; \left\Vert \vec\psi\right\Vert^{2}=d.
\end{equation}
The relation of Problem~\ref{prob:ortho.ellpos} to the balanced sum problem (Problem~\ref{prob:ellpos.balanced_sum}) is derived in the following Lemma.
\begin{lemma}\label{lem:ellpos.gap_or_no_gap}
  If the instance $\vec a$ of Problem~\ref{prob:ellpos.balanced_sum} allows for a balanced sum partition, then
  \begin{equation}
    \max_{\vec\psi\in\mathbb{S}^{d-1}_{d}}\left[f\left(\vec\psi\right)-C_{2}{\left( \vec a \cdot \vec\psi\right)}^{4}\right]
    = 2d^{2}-d \label{eq:ellpos.def_pi0}.
  \end{equation}
  On the other hand, if there is no such partition, we have
 \begin{align}
    \max_{\vec\psi\in\mathbb{S}^{d-1}_{d}}\left[f\left(\vec\psi\right)-C_{2}{\left( \vec a \cdot \vec\psi\right)}^{4}\right]
    &<&   \max_{\vec\psi\in\mathbb{S}^{d-1}_{d}}f\left(\vec\psi\right)\label{eq:ellpos.def_pi}\\
    &\leq&2d^{2}-d - \frac{2}{p(a d)}
    \label{eq:ellpos.def_pi2}.
  \end{align}
  where $p(x)=2 x^4$ is a non-negative polynomial.
\end{lemma}
For the sake of clarity we relegate the proof of the above lemma to the end of this section.
As a consequence of Lemma~\ref{lem:ellpos.gap_or_no_gap} the choice,
\begin{equation}
  C_{1}=2d^{2}-d-p{(a d)}^{-1},
  \label{eq:ellpos.choice}
\end{equation}
implies that an efficient algorithm deciding whether the inequality~\eqref{eq:ellpos.condition} is satisfied or not is also capable of deciding Prob.~\ref{prob:ellpos.balanced_sum} efficiently.
This is exactly the statement of Thm.~\ref{thm:ortho.hard}.\\

The last step we need to make is to find the parameters $R_{1}$ and $R_{2}$ leading to the choice~\eqref{eq:ellpos.choice}.
To this end, we set $R_{2}=\epsilon R_{1}$ with $0<\epsilon<1$ and introduce two positive parameters
\begin{equation}
  B_{1}=p{(a d)}^{-1},\qquad B_{2}=\frac{d a^2}{1+a^2}.
\end{equation}
Note that if $1\leq j\leq d$ is such that $|a_j|=\min_k |a_k|$, then for $\vec\psi^j$ given by $\psi^j_k=\sqrt{d} \delta_{jk}$ the function $f(\vec\psi^j)$ is equal to
\begin{equation}
f\left(\boldsymbol{\psi}^{j}\right)=\frac{d^{2}}{1+a^{2}}\left(1+a^{2}-2a_{j}^{2}\right).
\end{equation}
Since $a^2-2a_j^2\geq(d-2)a_j^2$ the quantity $f(\boldsymbol{\psi}^{j})$ is non-negative, so is the right hand side of Eq.~\eqref{eq:ellpos.def_pi}.
From~\eqref{eq:ellpos.def_pi2} we can find the bound
\begin{equation}
  \label{eq:ellpos.b1_bound}
  B_{1} \le d^{2}-d/2.
\end{equation}
Furthermore, $B_{2} \le d$.

Rearranging Eq.~\eqref{eq:ellpos.def_c1}, taking the square root and substituting~\eqref{eq:ellpos.choice} we can see that $R_1$ is implicitly defined by the relation
\begin{equation}
  \sqrt{2}\sqrt{\left(d^{2}-d-B_1\right)\left(1-\epsilon^{2}\right)+d\left(d-1\right)\epsilon^{2}}R_{1}=q_\pm.
  \label{eq:ellpos.equiv}
\end{equation}
If the left hand side of~\eqref{eq:ellpos.equiv} happens to be bigger than $1/2$, we need to take the $q_+$ solution on the right hand side (and $q_-$ in the opposite case). In order for the square roots in Eq.~\eqref{eq:ellpos.equiv} to be real-valued, we need to assume
\begin{equation}
  \left(d^{2}-d-B_1\right)\left(1-\epsilon^{2}\right)+d\left(d-1\right)\epsilon^{2}\geq0.
\end{equation}
and
\begin{equation}
  1-8R_{1}^{2}\left(1-\epsilon^{2}\right)B_{2}\geq0,\label{eq:ellpos.extra_cond}
\end{equation}
The latter condition assures that $q_\pm$ are real while the former condition, as it
does not depend on $R_{1}$, can be immediately solved for $\epsilon$:
\begin{equation}
  \epsilon^{2}\geq1-\frac{d\left(d-1\right)}{B_{1}}.
  \label{eq:ellpos.condition_epsilon1}
\end{equation}
However, Eq.~\eqref{eq:ellpos.condition_epsilon1} does not yield a universal bound for acceptable values of $\epsilon$ since $B_1$ depends on the particular instance $\vec a$.
To obtain a lower bound independent of $\vec a$, we use Eq.~\eqref{eq:ellpos.b1_bound}, obtaining:
\begin{equation}
\epsilon^2 \geq \frac{1}{2d - 1}.
\end{equation}
Since both sides of~\eqref{eq:ellpos.equiv} are non-negative, we can take the square of this relation and turn it it into a quadratic equation for $R_1$. Surprisingly, this equation has a trivial solution $R_1=0$ (only relevant while dealing with $q_-$) and a single  non-trivial solution which can be simplified to the form:
\begin{equation}
  R_{1}=\frac{1}{\sqrt{2}}\frac{\sqrt{d\left(d-1\right)-B_{1}\left(1-\epsilon^{2}\right)}}{d\left(d-1\right)-\left(B_{1}-B_{2}\right)\left(1-\epsilon^{2}\right)},
  \label{eq:ellpos_r1}
\end{equation}
The condition~\eqref{eq:ellpos.extra_cond} becomes trivially satisfied, while the left hand side of Eq. (\ref{eq:ellpos.equiv}) is greater than $1/2$ (relevant for $q_+$) for
\begin{equation}
  \epsilon^{2}\geq1-\frac{d\left(d-1\right)}{\left(B_{1}+B_{2}\right)}.
  \label{eq:ellpos.condtion_epsilon2}
\end{equation}
In the opposite case the inequality is reversed.
When~\eqref{eq:ellpos.condtion_epsilon2} occurs, we find that
\begin{align}
  q_{+} &= \frac{d\left(d-1\right)-B_{1}\left(1-\epsilon^{2}\right)}{d\left(d-1\right)-\left(B_{1}-B_{2}\right)\left(1-\epsilon^{2}\right)}\label{qp},\\
  q_{-} &=\frac{B_{2}\left(1-\epsilon^{2}\right)}{d\left(d-1\right)-\left(B_{1}-B_{2}\right)\left(1-\epsilon^{2}\right)},
\end{align}
while in the opposite case the parameters $q_{+}$ and $q_{-}$ swap.
These interrelations between the parameters imply that regardless of the validity of~\eqref{eq:ellpos.condtion_epsilon2}, the solution~\eqref{eq:ellpos_r1} uniquely determines  $q$ initially introduced in~\eqref{eq:ellpos.rho0} as given by the formula~\eqref{qp}. This parameter is manifestly smaller than $1$ and due to~\eqref{eq:ellpos.condition_epsilon1} it is also non-negative.
With the given choice of parameters (\ref{eq:ellpos_r1},~\ref{eq:ellpos.condtion_epsilon2}) and $q$ specified as above, we complete the reduction of the balanced sum problem to Prob.~\ref{prob:ortho.ellpos}.
To finalize the proof of Theorem~\ref{thm:ortho.hard}, we now state the proof of Lemma~\ref{lem:ellpos.gap_or_no_gap}.
\begin{proof}[Proof of Lemma~\ref{lem:ellpos.gap_or_no_gap}]
  The first part of the proof -- Eq.~\eqref{eq:ellpos.def_pi0} --  follows from a simple calculation utilizing the partition vector $\vec\psi$ defined in~\eqref{eq:ellpos.partition_vector}. Note that as $\vec a \cdot\vec \psi=0$, we immediately obtain the first equality in~\eqref{eq:ellpos.def_pi0}, which since $C_2$ is non-negative turns into inequality in~\eqref{eq:ellpos.def_pi}.

  To prove~\eqref{eq:ellpos.def_pi2}, we define the set of all possible ($2^d$ in total) partition vectors
  \begin{equation}
    \mathcal{Z} := \left\{ \vec z \in \BR^d\colon \forall i \, z_i = \pm 1 \right\}
  \end{equation}
  and (for an arbitrary $0<\lambda<1$) the set of vectors that are ``close'' to some element from $\mathcal{Z}$
  \begin{equation}
    \mathcal{B} := \left\{ \vec\psi\in\BR^d\colon \min_{\vec z\in \mathcal{Z}} \norm{\vec\psi - \vec z} \leq \frac{\lambda}{a} \right\}.
  \end{equation}
Because $a \geq 1$, the set $\mathcal{B}$ can be thought of as a disjoint union of $2^d$ balls centered around the elements of $\mathcal{Z}$.
For further convenience we denote $\tilde{\vec{z}} = \mathrm{argmin}_{\vec z\in\mathcal{Z}} \, \norm{\vec\psi - \vec z}$, and $\vec\delta := \vec\psi - \tilde{\vec{z}}$.
By construction $\tilde{z}_k=\mathrm{sign}\,\psi_k$ so that for all $k=1,\ldots,d$
\begin{equation}
\tilde{z}_k\delta_k=\tilde{z}_k\psi_k-\tilde{z}_k^2=|\psi_k|-1\geq-1.
\end{equation}
Since $\norm{\vec\psi}^2=d$ we find that
  \begin{equation}
    2\tilde{\vec{z}} \cdot\vec\delta = - \norm{\vec\delta}^2.
  \end{equation}
Using all the above, the fact that $\tilde{z}_k^2=1$ and  $\tilde{z}_k^3=\tilde{z}_k$, and the Jensen inequality we can further estimate
\begin{equation}
-\sum_{k=1}^d \psi_k^4\leq-d-\sum_{k=1}^d \delta_k^4\leq-d-\frac{\norm{\vec\delta}^4}{d}.
\end{equation}

As $\vec a$ does not allow for a balanced sum partition and both, $\tilde{\vec z}$ and $\vec a$ are integral, we must necessarily have $\vert \vec a \cdot \tilde{\vec z} \vert \geq 1$. Thus
\begin{equation}
1\leq\left|\boldsymbol{a}\cdot\tilde{\vec z}\right|=\left|\boldsymbol{a}\cdot\left(\boldsymbol{\psi}-\boldsymbol{\delta}\right)\right|\leq\left|\boldsymbol{a}\cdot\boldsymbol{\psi}\right|+\left|\boldsymbol{a}\cdot\boldsymbol{\delta}\right|\leq\left|\boldsymbol{a}\cdot\boldsymbol{\psi}\right|+a\norm{\boldsymbol{\delta}},
\end{equation}
so that
\begin{equation}
-\left|\boldsymbol{a}\cdot\boldsymbol{\psi}\right|\leq\min\left\{0,a\norm{\boldsymbol{\delta}}-1\right\},
\end{equation}
Taking all the above results together with $\left|\boldsymbol{a}\cdot\boldsymbol{\psi}\right|\leq a\norm{\boldsymbol{\psi}}=a\sqrt{d}$ we obtain
\begin{equation}
f(\vec \psi)\leq 2d^{2}-d-\frac{\norm{\vec\delta}^4}{d}+2 d^{3/2}a\frac{\min\left\{0,a\norm{\boldsymbol{\delta}}-1\right\}}{1+a^{2}}.
\end{equation}

We will now study two cases. For $\psi \in \mathcal{B}$, we have $0\leq\norm{\vec \delta}\leq\lambda/a$, so that
\begin{equation}
f(\vec \psi)\leq 2d^{2}-d-2 d^{3/2}a\frac{1-\lambda}{1+a^{2}},
\end{equation}
while for the opposite case ($\psi \notin \mathcal{B}$), when  $\norm{\vec \delta}>\lambda/a$, one finds
 \begin{equation}
f(\vec \psi)\leq 2d^{2}-d-\frac{\lambda^4}{d a^4}.
\end{equation}
Therefore, we have for any $\vec\psi \in \BR^d$ with $\norm{\vec\psi}^2 = d$
\begin{equation}
f(\vec \psi)\leq 2d^{2}-d- \min\left\{2 d^{3/2}a\frac{1-\lambda}{1+a^{2}},\frac{\lambda^4}{d a^4}\right\},
\end{equation}
so that by setting $\lambda=d^{-3/4}$ we obtain the desired result with $p(ad)=2{(ad)}^4$.
\end{proof}

\section{Proof of Lemma~\ref{lem:ortho.spheres}}
\label{sec:spheres}

To check whether a sphere with radius $R$ centered at $\estim\varrho$ is contained in the set of psd states, specialize Eq.~\eqref{eq:ellpos.worst_case} to the special case $R_1 = R_2$:
\begin{equation}
  \bra{\Psi}\estim{\varrho} \ket{\Psi} -R\sqrt{\sum_{i}v_{i}^2\left(\vec\psi\right)}\geq0.
 \label{eq:spheres.worst_case}
\end{equation}
Since for any pure state $ \ket{\Psi} $ the identity
\begin{equation}
  \sum_{i}v_{i}^{2}\left(\vec\psi\right)=\frac{2\left(d-1\right)}{d},\label{puresum-1}
\end{equation}
holds (Bloch vectors of pure states live on the hypersphere), the inequality in question becomes
\begin{equation}
  \bra{\Psi}\estim{\varrho} \ket{\Psi} -R\sqrt{\frac{2\left(d-1\right)}{d}}\geq0.
\end{equation}
Simple minimization with respect to $ \ket{\Psi} $ leads to the final result stated as Lemma~\ref{lem:ortho.spheres}.

\section{Proof of Theorem~\ref{thm:bayesian.hardness}}
\label{sec:proof_bayesian}

Let us now construct the polynomial time reduction of Prob.~\ref{prob:ortho.ellpos} to Prob.~\ref{prob:bayesian.cr}.
We will begin with the main observation of this reduction, namely Eq.~\eqref{eq:bayesian.criterion}.
\begin{lemma}\label{lem:bayesian.criterion}
  Let $\pi(\varrho)$ denote a Gaussian distribution on $\HermTrace$ and $\pi^+(\varrho) = C \pi(\varrho) \chi(\varrho)$ the corresponding restricted Gaussian with the same mean and covariance matrix, as defined in Eq.~\eqref{eq:bayesian.density_plus}.
  For any $\alpha \in [0,1]$, the credible ellipsoid $\Eps(r_\frac{\alpha}{C})$ with credibility $\frac{\alpha}{C}$ is contained in the psd if and only if the credible ellipsoid for $\pi^+$, $\Eps(r^+_{\alpha})$, with credibility $\alpha$ has the same radius, that is Eq.~\eqref{eq:bayesian.criterion} holds.
\end{lemma}
\begin{proof}
  The two cases of  $\Eps(r_\frac{\alpha}{C})$ being contained and not being contained in the psd states are illustrated in Fig.~\ref{fig:bayesian.ellipsoids}.
  First, assume that $\Eps(r_\frac{\alpha}{C}) \subset \States$, then
  \begin{equation}
    \frac{\alpha}{C} = \int_{\Eps(r_\frac{\alpha}{C})} \pi(\varrho) \, \dd \varrho.
    \implies
    \alpha = \int_{\Eps(r_\frac{\alpha}{C}) \cap \States } C \pi(\varrho) \, \dd \varrho.
  \end{equation}
  Note that the right equation is exactly the defining Eq.~\eqref{eq:bayesian.radius_plus} for the positive radius $r^+_{\alpha}$ if $r^+_\alpha = r_\frac{\alpha}{C}$.

  Now, assume that a part of the ellipsoid $O = \Eps(r_\frac{\alpha}{C})  \setminus \States \neq \emptyset$ lies outside the psd states.
  Then, as can be seen on the right side of Fig.~\ref{fig:bayesian.ellipsoids}, we need to enlarge $r^+_{\alpha}$ to compensate for the lost probability weight of $O$.
  The latter cannot be vanishing, since the Gaussian density $\pi(\varrho)$ is strictly positive.
  Therefore, $r^+_\alpha > r_\frac{\alpha}{C}$ in this case.
\end{proof}
Of course, the difference between $r_\frac{\alpha}{C}$ and $r^+_{\alpha}$ may in general become too small to be efficiently detectable.
However, we will show that for the instances of the balanced sum problem encoded in Problem~\ref{prob:ortho.ellpos}, this is not the case.
A first step toward this is the following Lemma.

\begin{lemma}\label{lem:bayesian.positivity_violation}
  Let $\vec a \in \mathbb{N}^d$ describe an instance of the balanced sum problem and
  \begin{equation}
    \label{eq:bayesian.positivity_violation.ellipsoid}
    \Eps_{\vec{a}} = \left\{ \varrho_0 + R_1 \sum_{i=1}^{i_d} u_i \sigma_i^ + R_2 \sum_{i=i_d + 1}^{d^2 - 1} u_i \sigma_i \colon \norm{\vec u}_2 = 1 \right\}
  \end{equation}
  the corresponding encoding ellipsoid for Problem~\ref{prob:ortho.ellpos} defined in~\ref{sec:ellpos}.
  There exists a polynomial $\tilde p$ such that if $\Eps_{\vec{a}}$ is not a subset of $\States$, there is an element $\varrho \in \Eps_{\vec{a}}$ with
  \begin{equation}
    \mathrm{mineig}(\varrho) \le -\tilde p{(\norm{\vec a})}^{-1} < 0.
  \end{equation}
\end{lemma}
\begin{proof}
  The main proof idea is to trace back the proof for polynomial gap in Lemma~\ref{lem:ellpos.gap_or_no_gap}.
  Recall that Eqs.~\eqref{eq:ellpos.def_pi0} and~\eqref{eq:ellpos.choice} ensure that if $\vec a$ has a balanced sum partition, there is a $\vec\Psi \in {\{\pm 1\}}^d$ such that $\vec a \cdot \vec\Psi = 0$ and
  \begin{equation}
    d^2 - \sum_k \psi_k^4 + {\left( d - \frac{{(\vec a \cdot \vec \psi)}^2}{1 + \norm{\vec a}^2} \right)}^2 - C_2 {(\vec a \cdot \vec \psi)}^4 = C_1 + p{(\norm{\vec a})}^{-1}.
  \end{equation}
  By tracing back the steps which lead to this equation, we find for $\ket{\Psi} := \sum_{k=1}^d \psi_k / \sqrt{d} \ket{k}$
  \begin{align}
    \label{eq:bayesian.posviol_1}
    \frac{2 (R_1^2 - R_2^2)}{d} \, p{(\norm{\vec a})}^{-1} + \bra{\Psi} \varrho_0 \ket{\Psi}^2 \\
    = R_1^2 \sum_i {\left( \bra\Psi \sigma_i^{(x)} \ket\Psi \right)}^2 + R_2^2 \sum_i {\left( \bra\Psi \sigma_i^{(y,z)} \ket\Psi \right)}^2 \\
    =: \sum_i R_i^2 {\left( \bra\Psi \sigma_i \ket\Psi \right)}^2
  \end{align}
  Due to the special choice for $\varrho_0$ in~\eqref{eq:ellpos.rho0} and $\vec a \cdot \vec \psi = 0$, we have
  \begin{equation}
    \bra\Psi\varrho_0\ket\Psi = \frac{q}{d}
  \end{equation}
  with $q$ defined in~\eqref{eq:ellpos.q}.
  Therefore, we can rewrite Eq.~\eqref{eq:bayesian.posviol_1} as
  \begin{align}
    \bra\Psi \varrho_0 \ket\Psi - \sqrt{\sum_i R_i^2 \bra\Psi \sigma_i \ket\Psi^2}
    &&= \frac{q}{d} \left( 1 - \sqrt{1 + \frac{2d (R_1^2 - R_2^2)}{q^2 \, p(\norm{\vec a})}} \right) \nonumber\\
    &&\le - \min \left( \frac{R_1^2 - R_2^2}{2 q \, p(\norm{\vec a})},\, \frac{2q}{d} \right)
    \label{eq:bayesian.posviol_2}
  \end{align}
  where we have used
  \begin{equation}
    1 - \sqrt{1 + x^2} \le
    \left\{
      \begin{array}{ll}
        -x^2 / 4 \quad &  x \le 2 \sqrt{2} \\
       -2 & x > 2 \sqrt{2} \\
      \end{array}
    \right.
  \end{equation}
  Since all the constants on the right hand side of Eq.~\eqref{eq:bayesian.posviol_2} can be expressed as polynomials in the input, it defines the polynomial $\tilde p(\norm{\vec a})$ of the lemma.
  The left hand side of that equation is equal to $\bra\Psi \varrho \ket\Psi$, where
  \begin{equation}
    \varrho = \varrho_0 + \sum_i R_i u_i \sigma_i \in \Eps_{\vec{a}}
  \end{equation}
  for the special choice of $u$ from~\eqref{eq:ellpos.positivity}.
  The claim of the lemma follows for this $\varrho$ using Eq.~\eqref{eq:bayesian.posviol_2}.
\end{proof}

We will now show how the explicitly parameterized ellipsoid~\eqref{eq:bayesian.positivity_violation.ellipsoid} can be encoded as a MVCR-ellipsoid of a Gaussian distribution.

\begin{lemma}\label{lem:bayesian.encoded_ellipsoid}
  Denote by
  \begin{equation}
   \Eps^* = \left\{ \varrho_0 + \sum_{i=1}^{d^2 - 1} u_i R_i \sigma_i \colon \norm{\vec u}_2 = 1 \right\}
  \end{equation}
  an ellipsoid $\Eps^* \subset \HermTrace$, which is axis-aligned with the coordinate axes defined by the generalized Pauli operators.

  Then, $\Eps^*$ can be encoded as a $\frac{\alpha}{C}$ MVCR-ellipsoid for a Gaussian distribution with mean $\varrho_0 \in \States$ and covariance matrix $\Sigma$.
  The latter is diagonal in the generalized Bloch basis $\sigma_i$ with entries $\Sigma_{ij} = R_i^2 \delta_{ij}$ and for the corresponding radius we have $r_\frac{\alpha}{C} = \sqrt{2}$.
  Hence, the credibility is given by
  \begin{equation}
    \label{eq:bayesian.encoded_ellipsoid.credibility}
    \alpha = C\, P\left(\Nhalf, 1\right),
  \end{equation}
  which can be calculated efficiently up to exponential precision for given $C$ and $N$.
 \end{lemma}
 \begin{proof}
   Since the generalized Pauli operators form an orthogonal system with $\tr \sigma_i \sigma_j = 2 \delta_{ij}$, we find for $\varrho \in \Eps^*$
   \begin{equation}
      \norm{\varrho}^2_2 = \sum_{i,j} u_i u_j \, R_i R_j \, {(\Sigma^{-1})}_{ij} \, 2 \delta_{ij} = 2 \norm{\vec u}_2^2.
   \end{equation}
   Therefore, $\Eps^* = \Eps(\sqrt{2})$ with mean $\varrho_0$ and the stated covariance matrix.
   The efficient computation of the credibility~\eqref{eq:bayesian.encoded_ellipsoid.credibility} is given later in the proof of Lemma~\ref{lem:bayesian.always_contained}.
 \end{proof}

\begin{figure}[t]
  \centering
  \begin{tikzpicture}[
    elip/.style={fill=white, line width=0},
    truncated/.style={pattern=north west lines, pattern color=blue},
    truncelip/.style={dashed, line width=2pt, fill=blue},
    rest/.style={pattern=north west lines, pattern color=green},
    psd/.style={dashed, color=red, line width=2pt}
  ]
    \clip (-4,-3) rectangle (5,3);

    \def\origelip{(0,0) ellipse [x radius=1.5, y radius=2, rotate=-30]}
    \def\truncatedelip{(0,0) ellipse [x radius=2, y radius=2.66666, rotate=-30]}

    \def\psdradius{5.5}
    \def\psdcircle{(psdcenter) [draw=none] circle (\psdradius)}
    \fill[rest]\truncatedelip;
    \draw[truncelip]\origelip;
    \draw\truncatedelip;

    \node at (-1,-4) (psdcenter) {};
    \node at (3,1) {\red{psd}};

    \draw[psd] ([shift=(50:\psdradius)]psdcenter) arc (50:100:\psdradius);
    \begin{scope}
      \clip \psdcircle;
      \fill[truncated]\truncatedelip;
      \fill[elip] \origelip;
    \end{scope}

    \node at (0,0) {$\Eps(r_{\frac{1-\alpha}{C}})$};
    \node at (2.8,-.4) {$\Eps(r^+_{1-\alpha})$};
  \end{tikzpicture}
  \caption{\label{fig:bayesian.r_separation}
    Same as Fig.~\ref{fig:bayesian.ellipsoids} (right).
    Note that the solid blue and hatched blue regions need to have the same volume.
  }
\end{figure}
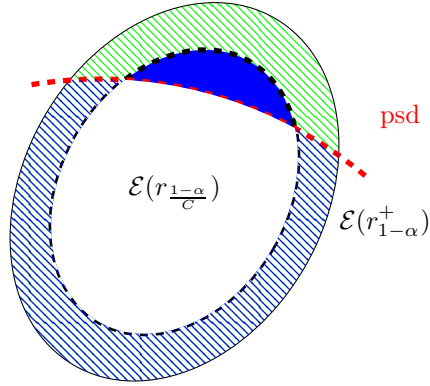

Based on the gap proven in Lemma~\ref{lem:bayesian.positivity_violation}, we will now turn to the following question:
In case Eq.~\eqref{eq:bayesian.criterion} does not hold -- that is the corresponding ellipsoid is not fully contained in the psd states -- is the corresponding gap always large enough to be efficiently detectable?
\begin{lemma}\label{lem:bayesian.r_separation}
   Let $\vec a \in \mathbb{N}^d$ be an instance of the balanced sum problem and denote by $\Eps_{\vec{a}}$ the corresponding encoding ellipsoid as given by Eq.~\eqref{eq:bayesian.positivity_violation.ellipsoid}.
  Furthermore, denote by $\pi_{\varrho_0,\Sigma}$ the Gaussian density, which encodes $\Eps_{\vec{a}} = \Eps(r_\frac{\alpha}{C})$ as an $\frac{\alpha}{C}$ credible region as given by Lemma~\ref{lem:bayesian.encoded_ellipsoid}.
  Assume that $\vec a$ has a balanced sum partition and, therefore, $\Eps_{\vec{a}}$ is not a subset of $\States$.

  Then, there exists a polynomial $p$ such that
  \begin{equation}
    {r^+_{\alpha}}^2 - {r_\frac{\alpha}{C}}^2 \ge 2^{-p(\log \norm{\vec a}_1)}.
  \end{equation}
  Here, $\norm{\vec a_1} = \sum_k \abs{a_k}$.
  In words, the gap of violation of Eq.~\eqref{eq:bayesian.criterion} can only become polynomially small in the logarithm of the size of the problem specification.
\end{lemma}
\begin{proof}
  First, let us lower bound the volume of $\Eps(r_\frac{\alpha}{C})$ that lies outside the psd states (the solid blue region in Fig.~\ref{fig:bayesian.r_separation}).
  From Lemma~\ref{lem:bayesian.positivity_violation} we know, that there exists a $\varrho\in \Eps(r_\frac{\alpha}{C})$ with smallest eigenvalue smaller than $- \tilde p{(\norm{\vec a})}^{-1}$ for some polynomial $\tilde p$.
  This also gives us a lower bound on
  \begin{equation}
    \mathrm{dist}(\varrho,\States) = \inf_{\varrho' \in \States} \norm{\varrho - \varrho'}_2.
  \end{equation}
  From~\cite[Theorem~III.2.8]{Bhatia_1997_Matrix} we know that for every $\varrho_+ \in \States$ the following bound holds:
  \begin{equation}
    \begin{split}
      \norm{\varrho - \varrho_+}_2
      &\ge \norm{\varrho - \varrho_-}_\infty
      \ge \norm{\vec\lambda^\uparrow(\varrho) - \vec\lambda^\uparrow(\varrho_+)}_2 \\
      &\ge \abs{\mineig(\varrho) - \mineig(\varrho_+)}
      \ge \tilde p{(\norm{\vec a})}^{-1}.
    \end{split}
  \end{equation}
  Therefore,
  \begin{equation}
    \label{eq:bayesian.r_separation.dist}
    \mathrm{dist}(\varrho,\States) \ge \tilde p{(\norm{\vec a})}^{-1}.
  \end{equation}
  This allows us to lower bound the volume of $\Eps(r_\frac{\alpha}{C})$ that lies outside the psd states by an ellipsoid with the same covariance, but radius ${(2\, \tilde p(\norm{\vec a}) \, \maxeig(\Sigma))}^{-1}$
  \begin{align}
    \label{eq:bayesian.r_separation.volume}
    \mathrm{Vol}\left( \Eps(r_\frac{\alpha}{C}) \setminus \States \right)
    &\ge \frac{\pi^{\Nhalf} \abs{\Sigma}}{\Gamma(\Nhalf + 1)} \, \frac{1}{{\left( 2 \tilde p(\norm{\vec a}) \, \maxeig(\Sigma) \right)}^{N}} \\
  \end{align}
  Furthermore, we have
  \begin{equation}
    \label{eq:bayesian.r_separation.volume2}
    \mathrm{Vol}\left(\Eps(r^+_{1-\alpha}) \setminus \Eps(r_\frac{1-\alpha}{C}) \right)
    = \mathrm{Vol}\left( \Eps(r_\frac{1-\alpha}{C}) \setminus \States \right)
  \end{equation}
  since the solid blue and hatched blue regions in Fig.~\ref{fig:bayesian.r_separation} must be of same size.

  We now relate the volume inequality~\eqref{eq:bayesian.r_separation.volume} to a lower bound for the Gaussian volume:
  Due to the set of states $\States$ having finite radius $\sqrt{\tfrac{2(d-1)}{d}}$~\cite[Eq.~(18)]{Kimura_2003_Bloch}, we must have $r^+_{\alpha} \le 2 \sqrt{2}$.
  Therefore,
  \begin{align}
    P\left( \Nhalf, \tfrac{{r^+_{\alpha}}^2}{2} \right) - P\left( \Nhalf, \tfrac{{r_\frac{\alpha}{C}}^2}{2} \right)
    &= \frac{1}{ {(2\pi)}^{\frac{N}{2}} \, \abs{\Sigma}^{\frac{1}{2}} } \, \int_{\Eps(r^+_{\alpha}) \setminus \Eps(r_\frac{\alpha}{C})} \mathrm{e}^{-\frac{1}{2} \norm{\varrho - \varrho_0}^2} \dd^N \varrho \\
    &\ge \frac{\mathrm{e}^{-4}}{{(2\pi)}^{\frac{N}{2}} \, \abs{\Sigma}^{\frac{1}{2}}} \, \mathrm{Vol}\left(  \Eps(r^+_{\alpha}) \setminus \Eps(r_\frac{\alpha}{C})  \right) \\
    &\ge \frac{\mathrm{e}^{-4} \pi^{\Nhalf} \, \abs{\Sigma}^{\frac{1}{2}}}{2^\frac{N}{2} \Gamma(\Nhalf + 1)} \, \frac{1}{{\left( 2 \tilde p(\norm{\vec a}) \, \maxeig(\Sigma) \right)}^{N}} \\
    &=: 2^{-p(\log \norm{\vec a}_1) - 1}
    \label{eq:bayesian.r_separation.p_diff}
  \end{align}
  Finally, note that the following crude inequality
  \begin{equation}
     P\left( \Nhalf, \tfrac{{r^+_{\alpha}}^2}{2} \right) - P\left( \Nhalf, \tfrac{{r_\frac{\alpha}{C}}^2}{2} \right)
     = \int_y^x \frac{t^{\Nhalf - 1} \ee^{-t}}{\Gamma(\Nhalf + 1)} \,\dd t \le x - y
  \end{equation}
  holds for $x \ge y$, since the integrand is less than 1.
  Therefore, with Eq.~\eqref{eq:bayesian.r_separation.p_diff}
  \begin{equation}
    {r^+_{\alpha}}^2 - {r_\frac{\alpha}{C}}^2 \ge 2^{-p(\log \norm{\vec a}_1)},
  \end{equation}
  which proofs the claim.
\end{proof}

We now turn to the problem of computing the normalization constant $C$ for the restricted Gaussian distribution~\eqref{eq:bayesian.density_plus}.
First, we efficiently compute a credibility $\alpha' \in [0,1]$ such that the corresponding credible ellipsoid $\Eps(r_\frac{\alpha'}{C})$ is guaranteed to be contained in the psd states without knowing the value of $C$.
This allows us to leverage Eq.~\eqref{eq:bayesian.criterion} to compute $C$.

\begin{lemma}\label{lem:bayesian.always_contained}
  Let $\vec a \in \mathbb{N}^d$ be an instance of the balanced sum problem and denote by $\Eps_{\vec{a}}$ the corresponding encoding ellipsoid as defined by Eq.~\eqref{eq:bayesian.positivity_violation.ellipsoid}.
  Denote by $\pi_{\varrho_0,\Sigma}$ the Gaussian density, which encodes $\Eps_{\vec{a}}$ as an $\alpha$ credible region according to Lemma~\ref{lem:bayesian.encoded_ellipsoid}.
  Then, the ellipsoid $\Eps(r)$ is fully contained in the psd states provided
  \begin{equation}
    \label{eq:bayesian.always_contained}
    r \le \sqrt{\frac{d}{2(d-1)}} \, \frac{\mineig\varrho_0}{\sqrt{\maxeig\Sigma}}
  \end{equation}
\end{lemma}
\begin{proof}
  We know that for any $\varrho \in \Eps(r)$ with $r$ fulfilling~\eqref{eq:bayesian.always_contained} the following inequalities hold
  \begin{align*}
    \norm{\varrho - \varrho_0}
    &\le \frac{1}{\sqrt{\mineig\Sigma^{-1}}}\, \norm{\varrho - \varrho_0}_\Sigma \\
    &\le \frac{1}{\sqrt{\mineig\Sigma^{-1}}}\, r \\
    &\le \sqrt{\frac{d}{2(d-1)}}\, \mineig \varrho_0
  \end{align*}
  since $\mineig\Sigma^{-1} = {(\maxeig \Sigma)}^{-1}$.
  Therefore, $\Eps(r) \subset \States$ due to Lemma~\ref{lem:ortho.spheres}.
\end{proof}

\begin{lemma}\label{lem:bayesian.normalization_constant}
  Using the same notation as Lem.~\ref{lem:bayesian.always_contained} and assuming Prob.~\ref{prob:bayesian.trucated_cr} can be solved efficiently.
  Then, for every instance $a$ of the balanced sum problem and the corresponding $\varrho_0, \Sigma$, we can efficiently approximate the normalization constant $C$ of $\pi^+_{\varrho_0,\Sigma}$ with exponentially small error.
  More precisely, we have
  \begin{equation}
    C = \tilde C (1 + \epsilon),
  \end{equation}
  where $\tilde C$ can be computed in polynomial time making the correction term $\epsilon$ exponentially small.
\end{lemma}
\begin{proof}
  Due to Lemma~\ref{lem:bayesian.always_contained} and $\mineig\varrho_0 > 0$, we can always find an $r > 0$ such that $\Eps(r)$ is fully contained in the psd.
  Indeed, the eigenvalues of $\varrho_0$ and $\Sigma$ are readily calculated because of their particular simple form in Eq.~\eqref{eq:ellpos.rho0} and Lemma~\ref{lem:bayesian.encoded_ellipsoid}:
  \begin{equation}
    \sqrt{\frac{d}{2(d-1)}} \, \frac{\mineig\varrho_0}{\sqrt{\maxeig\Sigma}}
    =  \frac{q}{R_1 \sqrt{2d(d-1)}}
  \end{equation}
  Set\footnote{%
    Note that $\alpha$ does not denote the credibility used for encoding the ellipsoid in question, but an auxiliary ellipsoid used for computing $C$ here.
  }
  \begin{equation}
    \alpha := P\left( \Nhalf, \tfrac{r^2}{2} \right).
  \end{equation}
  Since we can choose $r$ as small as we want, we may assume that $x = \frac{r^2}{2} \ll 1 < \Nhalf$.
  In this regime, we can expand the normalized incomplete $\Gamma$-function $P$ in a power series~\cite{Gil_2012_Efficient}
  \begin{equation}
    \label{eq:bayesian.normalization_constant.incomplete_gamma}
    P\left( \Nhalf, x \right) = \frac{x^{\Nhalf} \ee^{-x}}{\Gamma\left( \Nhalf + 1 \right)} \sum_{k=0}^\infty \frac{x^k}{{\left( \Nhalf + 1 \right)}_k},
  \end{equation}
  where
  \begin{equation}
    {\left( \Nhalf + 1 \right)}_k = \frac{\Gamma\left( \Nhalf + k + 1 \right)}{\Gamma\left( \Nhalf + 1 \right)}.
  \end{equation}
  Truncating the series in Eq.~\eqref{eq:bayesian.normalization_constant.incomplete_gamma} for $k \ge k_0$
  \begin{equation}
    \label{eq:bayesian.normalization_constant.series}
    P\left( \Nhalf,x \right) = P_{k_0}\left( \Nhalf,x \right) + R_{k_0}\left( \Nhalf,x \right),
  \end{equation}
  with
  \begin{equation}
    P_{k_0}\left( \Nhalf,x \right)
    = \frac{x^{\Nhalf} \ee^{-x}}{\Gamma\left( \Nhalf + 1 \right)} \sum_{k=0}^{k_0} \frac{x^k}{{\left( \Nhalf + 1 \right)}_k}
  \end{equation}
  we can derive a bound on the truncation error $R_{k_0}(\Nhalf,x)$~\cite[Eq.~(2.18)]{Gil_2012_Efficient}
  \begin{equation}
    R_{k_0}(\Nhalf,x) \le \frac{x^{\Nhalf + k_0} \ee^{-x}}{\Gamma(\Nhalf + k_0 + 1)}\, \frac{\Nhalf + k_0}{\Nhalf + k_0 - x - 1}.
  \end{equation}
  Since $x \ll 1$, the term $x^{k_0}$ ensures that we can make the error in computing $\alpha$ exponentially small using only polynomial time in evaluating $P_{k_0}(\Nhalf, x)$.\\

  Assume that we have computed $\tilde\alpha = \alpha-\epsilon$ for some truncation error $\epsilon = R_{k_0}(\Nhalf,x) > 0$.
  We may now use the (postulated) efficient algorithm for Prob.~\ref{prob:bayesian.trucated_cr} to compute the radius of the manifestly positive MVCR $r^+_{\tilde\alpha}$ and, hence, using Eq.~\eqref{eq:bayesian.criterion} the normalization constant:
  Since $C>1$, we have with $r_{\alpha} = r$
  \begin{equation}
    r_\frac{\tilde\alpha}{C} = r_\frac{\alpha-\epsilon}{C} < r_{\alpha} \implies \Eps(r_\frac{\tilde\alpha}{C}) \subset \States \implies r_\frac{\tilde\alpha}{C} = r^+_{\tilde\alpha} \le r_{\alpha}.
  \end{equation}
  Therefore, the ellipsoid with radius $r^+_{\tilde\alpha}$ is also contained in the psd states.
  The same holds true if we replace $r^+_{\tilde\alpha}$ by the actual output $r^+_{\tilde\alpha} \pm \delta$ of the postulated efficient algorithm for Prob.~\ref{prob:bayesian.cr}
  Here, $\delta$ denotes the (selectable) accuracy.
  By choosing $\delta$ small enough and possibly replacing the original radius $r$ by $r - \delta$, we can ensure that
  \begin{equation}
    \label{eq:bayesian.normalization_constant.small_enough_r}
    \Eps(r^+_{\tilde\alpha} \pm \delta) \subset \States,
  \end{equation}
  as well.
  Therefore, Eq.~\eqref{eq:bayesian.criterion} holds and we find
  \begin{align}
    \label{eq:bayesian.normalization_constant.almost_c}
    \frac{\tilde\alpha}{C}
    &= P\left( \Nhalf, \tfrac{{r^+_{\tilde\alpha}}^2}{2} \right) \\
    &= P\left( \Nhalf, \tfrac{{(r^+_{\tilde\alpha} \pm \delta)}^2}{2} \right) - \frac{1}{\Gamma(\Nhalf)} \int_{\tfrac{{r^+_{\tilde\alpha}}^2}{2}}^{\tfrac{{(r^+_{\tilde\alpha} \pm \delta)}^2}{2}}\, t^{\Nhalf - 1} \ee^{-t} \dd t.
  \end{align}
  The first addend on the right hand side can be evaluated using the same series expansion as in Eq.~\eqref{eq:bayesian.normalization_constant.series}, since we are in the same regime $\tfrac{{r^+_{\tilde\alpha}}^2}{2} \ll \Nhalf$.
  The second addend can be bounded by
  \begin{equation}
    \label{eq:bayesian.normalization_constant.upper_bound}
    \abs{\frac{1}{\Gamma(\Nhalf)} \int_{\tfrac{{r^+_{\tilde\alpha}}^2}{2}}^{\tfrac{{(r^+_{\tilde\alpha} \pm \delta)}^2}{2}}\, t^{\Nhalf - 1} \ee^{-t} \dd t}
    < \frac{\left( 2 {r^+_{\tilde\alpha}} \delta + \delta^2 \right)}{2}
  \end{equation}
  since
  \begin{equation}
    \frac{t^{\Nhalf - 1} \ee^{-t}}{\Gamma(\Nhalf)} < 1.
  \end{equation}
  Let us assume w.l.o.g.\ $r^+_{\tilde\alpha} \le 1$.
  This bound, as well as the error bound $\epsilon' > 0$ for the finite series-evaluation of $P$ in~\eqref{eq:bayesian.normalization_constant.almost_c} leads to
  \begin{equation}
    \frac{\tilde\alpha}{C} = P_{k_0}\left( \Nhalf, \tfrac{{(r^+_{\tilde\alpha} \pm \delta)}^2}{2} \right) + \epsilon' \pm D \delta
  \end{equation}
  for some appropriate constant $D$.
  A little arithmetic gives
  \begin{equation}
    \label{eq:bayesian.normalization_constant.formula_c}
    C = \frac{\tilde\alpha}{P_{k_0}(\ldots)} \, \left( 1 - \frac{\epsilon' \pm D\delta}{P_{k_0}(\ldots) + \epsilon' \pm D\delta} \right).
  \end{equation}
  By assumption we can make both $\epsilon'$ and $\delta$ exponentially small using only polynomial time while $P_{k_0}(\Nhalf, x) \uparrow P(\Nhalf,x)$ for $k_0 \to \infty$, the correction to
  \begin{equation}
    \tilde C = \frac{\tilde\alpha}{P_{k_0}\left( \Nhalf, \tfrac{{(r^+_{\tilde\alpha} \pm \delta)}^2}{2} \right)}
  \end{equation}
  in Eq.~\eqref{eq:bayesian.normalization_constant.formula_c} can be made exponentially small using polynomial time.
  On the other hand, $\tilde C$ can be computed in polynomial time as well.
\end{proof}

We now have all the necessary parts for the proof of the main theorem, which will conclude this section.

\begin{proof}[Proof of Thm.~\ref{thm:bayesian.hardness}]
  The proof follows the outline stated at the beginning of this section:
  First, we encode the ellipsoid of Problem~\ref{prob:ortho.ellpos} to be checked as a MVCR of a Gaussian with mean $\varrho_0$ and covariance matrix $\Sigma$ according to Lemma~\ref{lem:bayesian.encoded_ellipsoid}.
  Using Lemma~\ref{lem:bayesian.normalization_constant}, we compute an estimate $\tilde C$ to the normalization constant $C$.
  Using the techniques from the proof of the aforementioned Lemma, we may compute an estimate
  \begin{equation}
    \alpha = C \, P\left( \Nhalf, 1 \right) = \tilde C (1 + \epsilon) \left( P_{k_0}\left( \Nhalf, 1 \right) + \epsilon' \right) = \tilde\alpha + \epsilon''.
  \end{equation}
  This can be done for exponential small errors $\epsilon, \epsilon'$ in polynomial time.
  Here, the computable value is given by
  \begin{equation}
    \tilde\alpha = \tilde C \, P_{k_0}\left(\Nhalf, 1 \right).
  \end{equation}
  An exponential small difference of $\alpha$ and $\tilde\alpha$ also implies an exponential small difference of $r^+_{1-\alpha}$ and $r^+_{\tilde\alpha}$:
  Set $x := r^+_{\alpha}$ and $\tilde x := r^+_{\tilde\alpha}$ and assume $x > \tilde x$ -- the opposite case can be treated along the same lines by choosing a larger constant as a bound for $\tilde x$.
  Following Eq.~\eqref{eq:bayesian.r_separation.p_diff}, we have
  \begin{align*}
    P\left( \Nhalf, \tfrac{x^2}{2} \right) - P\left( \Nhalf, \tfrac{{\tilde x}^2}{2} \right)
    &\ge \frac{\mathrm{e}^{-4}}{{(2\pi)}^{\frac{N}{2}} \, \abs{\Sigma}^{\frac{1}{2}}} \, \mathrm{Vol}\left(  \Eps(x) \setminus \Eps(\tilde x)  \right) \\
    &= \frac{\mathrm{e}^{-4}}{2^{\Nhalf} \Gamma(\Nhalf + 1)} \left( x^N - {\tilde x}^N \right).
  \end{align*}
  Since for fixed $N$, the left hand side can be made exponentially small in polynomial time by improving $\tilde\alpha$, so can the right hand side.
  Therefore, the difference $\abs{x - \tilde x}$ can be made exponentially small as well.

  Now, choose the errors $\epsilon$ and $\epsilon'$ in such a way that
  \begin{equation}
    \abs{r^+_{\alpha} - r^+_{\tilde\alpha}} \le \frac{\Delta}{4}.
  \end{equation}
  Here, $\Delta = 2^{-p(\log \norm{\vec a}_1)}$ is the (at worst exponentially small) gap from Lemma~\ref{lem:bayesian.r_separation}.
  Furthermore, we run the algorithm for computing $r^+_{\tilde\alpha}$ with precision $\delta = \frac{\Delta}{4}$ and denote the result by $\tilde r$.
  If $\abs{\tilde r - \sqrt{2}} \le \frac{\Delta}{2}$, we know that $r^+_{\alpha} = r_\frac{\alpha}{C}$ and the ellipsoid is fully contained in the psd states.
  Otherwise we know that it is not.
\end{proof}
\end{document}